\documentclass[11pt]{article}
\usepackage{amsmath,amsthm,amssymb,amsfonts}
\usepackage{mathtools}
\usepackage{mathrsfs}
\usepackage{multirow}
\usepackage[inline]{enumitem}
\usepackage{hyperref}
\usepackage{bm}
\usepackage[margin=1in]{geometry}
\usepackage{colortbl}
\usepackage{booktabs}
\usepackage[mathscr]{euscript}
\usepackage{microtype}

\usepackage[noblocks]{authblk}

\usepackage{sectsty}
\usepackage{setspace}
\usepackage{caption}
\usepackage{todonotes}

\newtheorem{assumption}{Assumption}
\newtheorem{theorem}{Theorem}

\newtheorem{proposition}{Proposition}
\theoremstyle{remark}

\usepackage[round]{natbib}
 \bibpunct[, ]{(}{)}{,}{a}{}{,}%

\input{latexmacro}

\DeclareMathOperator{\IMSE}{IMSE}
\newcommand{\SFM}{\mathsf{M}}
\newcommand{\SFW}{\mathsf{W}}
\newcommand{\SFY}{\mathsf{Y}}
\newcommand{\SFZ}{\mathsf{Z}}
\newcommand{\CalL}{\mathcal{L}}
\newcommand{\CalR}{\mathcal{R}}

\begin{document}
%%%%%%%%%%%%%%%%

\title{Stochastic Kriging for Inadequate Simulation Models}

\author{Lu Zou, Xiaowei Zhang}
\affil{Department of Industrial Engineering and Decision Analytics, The Hong Kong University of Science and Technology, Clear Water Bay, Hong Kong}

% \author{\textbf{(Authors' names blinded for peer review)}}
\date{}

\maketitle

\begin{abstract}
Stochastic kriging is a popular metamodeling technique for representing the unknown response surface of a simulation model. However, the simulation model may be inadequate in the sense that there may be a non-negligible discrepancy between it and the real system of interest. Failing to account for the model discrepancy may conceivably result in erroneous prediction of the real system's performance and mislead the decision-making process. This paper proposes a metamodel that extends stochastic kriging to incorporate the model discrepancy. Both the simulation outputs and the real data are used to characterize the model discrepancy. The proposed metamodel can provably enhance the prediction of the real system's performance. We derive general results for experiment design and analysis, and demonstrate the advantage of the proposed metamodel relative to competing methods. Finally, we study the effect of Common Random Numbers (CRN). The use of CRN is well known to be detrimental to the prediction accuracy of stochastic kriging in general. By contrast, we show that the effect of CRN in the new context is substantially more complex. The use of CRN can be either detrimental or beneficial depending on the interplay between the magnitude of the observation errors and other parameters involved. 

\end{abstract}


\section{Introduction}

Simulation is used broadly in various areas including health care, finance, supply chain management, etc. to analyze the performance of complex stochastic systems. The popularity is attributed to the modeling flexibility that can account for virtually any level of details of the system and any performance measure of interest. However, simulation models are often computationally expensive to execute, which severely restricts the usefulness of simulation when timely decision making is necessary. Simulation metamodeling has been developed actively in the simulation community in order to alleviate this computational issue; see \cite{BartonMeckesheimer06} for an overview. The basic idea is that the user only executes the simulation model at a small number of carefully selected ``design points''. A metamodel is then built to approximate the response surface of the simulation model by interpolating the simulation outputs. The responses at other places are predicted by the metamodel without running the simulation at all, thereby reducing the computational cost substantially.

Kriging is a metamodeling technique that has been studied extensively in both the spatial statistics community \citep{stein1999interpolation} and the design and analysis of computer experiments community \citep{fang2006design}; see also \cite{kleijnen2009} for a review. Kriging imposes a spatial correlation structure on the unknown response surface, and thus can provide a good global fit over the design space of interest. Its analytical tractability and ease of use also contribute to its wide adoption. Stochastic kriging (SK) was introduced by \cite{ankenman2010stochastic} as an extension of kriging in the stochastic simulation setting to account for the uncertainty that results from the random simulation errors. The SK metamodel has drawn substantial attention from the simulation community in recent years. It has been successfully used to quantify the impact of input uncertainty on the simulation outputs \citep{BartonNelsonXie14,XieNelsonBarton14} as well as to guide the random search for the optimal design of a simulation model \citep{quan2013,sun2014}. Numerous efforts have also been devoted to understand its theoretical properties \citep{chen2012effects,chen2014stochastic} and to enhance its performance \citep{chen2013enhancing,QuFu14,ShenHongZhang17}.

Nevertheless, there may be a non-negligible discrepancy between the simulation model  and the real system of interest, in which case the model is said to be inadequate. This may occur in practice either because data collection is not sufficient to build an adequate model, or because certain detailed structure of the real system is overlooked. Model discrepancy is typically addressed as part of model validation and calibration in simulation literature; see, e.g., \citet[Chapter 10]{BanksCarsonNelsonNicol09} for an introduction. Specifically, model validation is concerned with comparing the simulation outputs with the observations of the real system via statistical tests or expert assessment \citep{Sargent13}. Model calibration, on the other hand, is the iterative process of comparing the model to the real system, collecting more data, and refining the model by adjusting the parameters and even the structure \citep{Xu2017}. The process of validation and calibration is time consuming due to the high computational cost of running a simulation model. Thus, it often stops when the time constraint for model development is met and does not necessarily end up with a high-fidelity model. The model discrepancy, regardless of its magnitude, will be normally neglected along with the observations of the real system  in the subsequent simulation analysis. 

Consequently, as a surrogate of the simulation model the SK metamodel would conceivably provide erroneous prediction about the real system and potentially mislead the system optimization, if the model discrepancy is significant but is discarded. The present paper attempts to address this issue. We consider a metamodel that extends SK to incorporate the model discrepancy in a coherent fashion and can combine both the simulation outputs and the observations of the real system to improve the prediction. We call the new metamodel stochastic kriging for inadequate simulation models (SK-i). We prove that by integrating both datasets, SK-i predicts the performance of the real system  with smaller mean squared error (MSE) than completing methods that reply on only one of the datasets. 

The SK-i metamodel represents the model discrepancy as a realization of a Gaussian random field, an idea proposed and popularized by \cite{kennedy2001bayesian} as a means for Bayesian calibration for deterministic simulation models. The approach is recently adopted by \cite{yuan2015calibration} and extended to the stochastic simulation setting. They propose a sequential procedure that aims to dynamically allocate the design points and meanwhile update the calibration parameter of the simulation model. There are several differences between the SK-i metamodel and the Bayesian calibration approach, however. First, the Bayesian calibration approach assumes the existence of an unknown parameter in the simulation model and the objective is to adjust its value to reduce the model discrepancy. By contrast, we follow the conventional setting of SK and assume that the calibration process has been completed but there may exist non-negligible model discrepancy. Second, the Bayesian calibration approach estimates the unknown parameters by computing their posterior distribution possibly via Markov chain Monte Carlo methods \citep[Part III]{GelmanCarlinSternDunsonVehtariRubin14}, whereas we follow the frequentist perspective and perform maximum likelihood estimation. The experiment designs developed in \cite{yuan2015calibration} and in this paper also reflect the different perspectives. Third, a critical assumption which the Bayesian calibration approach relies on is that the simulation errors are homoskedastic, i.e., the variance is constant at different design points. However, stochastic simulation models, especially those for queueing systems, usually have heteroskedastic simulation errors; see, e.g., \cite{cheng1999} and \cite{ShenHongZhang17}. Like SK, the methodology developed for SK-i is not restricted by such an assumption.

A distinctive characteristic of SK is that it involves two types of uncertainty -- one about the response surface of the simulation model and the other about the simulation errors -- and their interplay leads to various insights about the simulation experiment design. For instance, the optimal allocation of the simulation replications at a chosen design point is not simply proportional to the standard deviation of the simulation errors there, but is distorted by the spatial correlation structure imposed by SK; see \cite{ankenman2010stochastic}. In order to account for the model inadequacy, the SK-i metamodel involves two additional types of uncertainty -- one about the model discrepancy and the other about the errors in the observations of the real system. The four distinct types of uncertainty in SK-i and their interplay give rise to more sophisticated and even contrary results relative to SK. For example, the former appear in the experiment design, whereas the latter in the analysis of the effect of Common Random Numbers (CRN).

CRN is a variance reduction technique that is used widely due to its simplicity; see, e.g., \citet[Chapter V.6]{AsmussenGlynn07}. It is well known that in general, CRN increases the MSE of the SK predictor, albeit having beneficial effects in other aspects such as gradient estimation. This therefore precludes the use of CRN in conjunction with SK if the prediction matters the most; see \cite{ankenman2010stochastic}, \cite{chen2012effects}, and \cite{chen2013enhancing}. However, we show via both a stylized model and extensive numerical experiments that the effect of CRN is complex. It can be either beneficial or detrimental to the prediction. This is due to the presence of the two additional types of uncertainty in SK-i and their interplay with the two types of uncertainty in SK.

The contributions of the present paper are summarized as follows. First, we introduce the SK-i metamodel that extends the SK metamodel to the context where the inadequacy of the simulation model is non-negligible relative to the real system of interest. We provide a mathematical foundation for SK-i, addressing a series of problems that arise naturally in the new context with regard to, e.g., MSE-optimal prediction, parameter estimation, and experiment design. %For example, we derive the MSE-optimal predictor given the augmented data and propose a parameter estimation scheme based on maximum likelihood estimation. We also discuss the experiment design regarding the allocation of the simulation replications in order to minimize the integrated MSE of prediction. 

Second, we prove that SK-i yields more accurate prediction than competing methods, thanks to its capability of using both the simulation outputs and the observations of the real system jointly for prediction. This intuitive but important result has a  two-fold implication. On one hand, it indicates another usage of the real observations in addition to being used for calibrating the simulation model, i.e., they can and should be combined with the simulation outputs to improve prediction. On the other hand, it suggests that developing a simulation model, even a crude one, may help produce better predictions of the real system's performance than merely relying on the physical observations. 

Third, we analyze the effect of CRN on the prediction accuracy of the SK-i metamodel. We show that the situation is substantially more complex for SK-i than SK. In particular, the use of CRN may be beneficial to prediction under certain circumstances. The surprising result is essentially due to the presence of the additional uncertainty about the model discrepancy and the observation errors as well as their interplay with the uncertainty from the SK metamodel. 

The present paper extends our preliminary results in \cite{ZhangZou16} in numerous ways. It proves theoretically that SK-i yields better prediction than the competing methods which was demonstrated merely empirically in the previous work, a derivation of the experiment design that minimizes the integrated MSE, an in-depth analysis showing that the effect of CRN is much more complicated than the well-known prior result, and additional numerical experiments.

The remainder of the paper is organized as follows. In \S\ref{sec:model}, we introduce SK-i and derive the MSE-optimal predictor. In \S\ref{sec:SK-i}, we analyze SK-i in depth, proving that it has the smallest MSE among the three competing alternatives. We also analyze the effects of various parameters of SK-i  on its prediction accuracy. In \S\ref{sec:estimation}, we discuss parameter estimation and experiment design. In \S\ref{sec:experiments}, we illustrate the advantage of SK-i via numerical experiments. In \S\ref{sec:CRN}, we study the effect of CRN on the prediction accuracy via both a stylized model and numerical illustration. We conclude in \S\ref{sec:conclusion} and collect numerous technical results in the Appendices.

\section{Model Formulation}\label{sec:model}
We first review stochastic kriging, and then present the new metamodel to account for the  inadequacy of the simulation model relative to the real system.

\subsection{Stochastic Kriging}
Stochastic kriging (SK) was proposed in \cite{ankenman2010stochastic} as  a flexible, interpolation-based approach to modeling the relationship between the unknown response surface of a stochastic simulation model and the controllable design variables. Let $\BFx=(x_1,\ldots,x_d)^\intercal \in\Real^d$ denote the design variable and $\SFY(\BFx)$ denote the unknown response surface, i.e., the mean performance measure of the simulation model evaluated at $\BFx$. SK models the response surface as a realization of a Gaussian random field, namely,
\begin{equation}\label{eq:SK_GRF}
\SFY(\BFx) =  \BFf^\intercal(\BFx) \BFbeta+\SFM(\BFx),
\end{equation}
where $\BFf(\BFx)$ is vector of known functions, $\BFbeta$ is a vector of unknown parameters having the same dimensionality, and $\SFM$ is a zero-mean Gaussian random field. The  term $\BFf^\intercal(\BFx_i) \BFbeta$ represents the ``trend'' or mean of the response surface. Components of $\BFf(\BFx)$ can be domain-specific functions \citep{YangAnkenmanNelson07} or basis functions such as polynomials.

Given an experiment design $\{(\BFx_i, n_i):i=1,\ldots,k\}$, $n_i$ simulation replications are executed at each design point $\BFx_i\in\Real^d$. Let $y_j(\BFx_i)$ denote the simulation output from  replication $j$ at design point $\BFx_i$, for $i=1,\ldots,k$ and $j=1,\ldots,n_i$. Then, the output is expressed in SK as
\begin{equation}\label{eq:SK}
y_j(\BFx_i)=\SFY(\BFx_i)+\epsilon_j(\BFx_i)=\BFf^\intercal(\BFx_i) \BFbeta+\SFM(\BFx_i)+\epsilon_j(\BFx_i),
\end{equation}
where $\epsilon_j(\BFx_i)$ is the simulation error for replication $j$ taken at $\BFx_i$. % Notice that the uncertainty in $\epsilon_j(\BFx_i)$ is due to the nature of stochastic simulation and is referred to as \textit{intrinsic uncertainty}. By contrast, the randomness of $\SFM$ represents the uncertainty about the response surface and is imposed to facilitate statistical inference, thereby referred to as \textit{extrinsic uncertainty}.

A main purpose of SK is to predict the response at any arbitrary point $\BFx_0$ using the simulation outputs $\{y_j(\BFx_i):j=1,\ldots,n_i, i=1,\ldots,k\}$, instead of running additional simulation. We denote the sample mean of the simulation outputs and the simulation errors respectively by
\begin{equation}\label{eq:sample_avg}
\overline{y}(\BFx_i)\coloneqq\frac{1}{n_i}\sum_{j=1}^{n_i}y_j(\BFx_i)\quad\mbox{and}\quad \overline{\epsilon}(\BFx_i)\coloneqq\frac{1}{n_i}\sum_{j=1}^{n_i}\epsilon_j(\BFx_i),
\end{equation}
$i=1,\ldots,k$. Let $\Sigma_{\SFM}(\cdot, \cdot)$ denote the covariance function of $\SFM$, i.e., $\Sigma_{\SFM}(\BFx, \BFx') = \Cov(\SFM(\BFx), \SFM(\BFx'))$. Let $\BFSigma_{\SFM}$ denote the covariance matrix of $(\SFM(\BFx_1),\ldots,\SFM(\BFx_k))$ and $\BFSigma_\epsilon$ denote the covariance matrix of $(\overline{\epsilon}(\BFx_1),\ldots,\overline{\epsilon}(\BFx_k))$. In addition, let $\BFSigma_{\SFM}(\BFx_0,\cdot)$ denote the $k\times 1$ vector  whose $i^{\mathrm{th}}$ component is $\Cov(\SFM(\BFx_0), \SFM(\BFx_i))$, the spatial covariance between  the prediction point $\BFx_0$ and design point $\BFx_i$, $i=1,\ldots,k$. Assuming that $\BFSigma_{\SFM}$, $\BFSigma_\epsilon$, and $\BFbeta$ are known, the best linear unbiased predictor (BLUP) of $\SFY(\BFx_0)$ that minimizes the mean squared error (MSE) of the prediction is
\begin{equation}\label{eq:SK_BLUP}
\widehat\SFY(\BFx_0) =\BFf^\intercal(\BFx_0)\BFbeta+\BFSigma^\intercal_{\SFM}(\BFx_0,\cdot)  [\BFSigma_{\SFM}+\BFSigma_\epsilon]^{-1}( \overline{\BFy}-\BFF\BFbeta),
\end{equation}
where $ \overline{\BFy}\coloneqq (\overline{y}(\BFx_1),\ldots,\overline{y}(\BFx_k))^\intercal$ and $\BFF \coloneqq (\BFf(\BFx_1),\ldots,\BFf(\BFx_k))^\intercal $. The optimal MSE is
\begin{equation}\label{eq:SK_MSE}
\MSE^*(\widehat\SFY(\BFx_0)) = \Sigma_{\SFM}(\BFx_{0},\BFx_{0}) -  \BFSigma_{\SFM}^\intercal(\BFx_{0},\cdot)[\BFSigma_{\SFM}+\BFSigma_\epsilon]^{-1} \BFSigma_{\SFM}(\BFx_{0},\cdot).
\end{equation}

In practice, the covariance function $\Sigma_{\SFM}$ must be specified in advance. For example, a usual assumption is
$\Sigma_{\SFM}(\BFx,\BFx') = \tau_{\SFM}^2 \CalR_{\SFM}(\BFx, \BFx';\BFtheta_{\SFM})$, where $\tau_{\SFM}^2$ represents the spatial variance of $\SFM$ and $\CalR_{\SFM}$ is a correlation function with unknown parameter $\BFtheta_{\SFM}$ to be estimated. It is also usually assumed that $\SFM$ is second-order stationary, which means that $\CalR_{\SFM}$ depends on $(\BFx,\BFx')$ only through $\BFx-\BFx'$. A typical example is the squared exponential correlation function of the form $\CalR(\BFx, \BFx';\BFtheta) = \exp\left(\sum_{i=1}^d \theta_i(x_i - x'_i)^2\right)$, where $\BFtheta = (\theta_1,\ldots,\theta_d)^\intercal$; see \cite{XieNelsonStaum10} for a comparison of various correlation functions and their influence on SK.

\subsection{Stochastic Kriging for Inadequate Simulation Models}

The simulation model used to approximate the real system of interest may be \textit{inadequate}, meaning that the discrepancy between the simulation model and the real system is non-negligible. This occurs possibly because the data used for model construction is not sufficient, or because the real system is highly complex and the certain structural details are not incorporated in the simulation model. It is conceivable that using the SK metamodel as a surrogate of the simulation model but neglecting the issue of model discrepancy may lead to mis-informed, suboptimal decisions for the real system. In this section, we propose a new metamodel that captures simultaneously both the response surface of the simulation model and its model discrepancy.

Let $\SFZ(\BFx)$ denote the mean performance measure of the real system evaluated at $\BFx$. Suppose that
\begin{equation}\label{real_formulation}
\SFZ(\BFx)=\rho\SFY(\BFx)+\delta(\BFx),
\end{equation}
where $\rho$ is an unknown parameter and $\delta(\cdot)$ is referred to as the \textit{model discrepancy function}. This formulation is adopted from  \cite{kennedy2001bayesian}, which analyzes the calibration of a \textit{deterministic} simulation model against the real system. To build a metamodel that is compatible with SK, we represent the unknown model discrepancy function as
\begin{equation}\label{error_part}
\delta(\BFx)=\BFg^\intercal(\BFx)\BFgamma+\SFW(\BFx),
\end{equation}
where $\BFg(\BFx)$ is a vector of known functions, $\BFgamma$ is a vector of unknown parameters having the same dimensionality, and $\SFW$ is a zero-mean Gaussian random field with covariance function $\Sigma_{\SFW}$. In the light of \eqref{eq:SK_GRF} and \eqref{error_part}, the formulation \eqref{real_formulation} can be rewritten as
\begin{equation}\label{eq:metamodel}
\SFZ(\BFx)=\rho[\BFf^\intercal(\BFx)\BFbeta+\SFM(\BFx)]+\BFg^\intercal(\BFx)\BFgamma+\SFW(\BFx),
\end{equation}
which will henceforth be referred to as stochastic kriging for inadequate simulation models (SK-i).

Clearly, aside from the simulation outputs, observations of the real system are needed in order to to quantify the model discrepancy. Let $\{\BFt_i:i=1,\ldots,\ell\}$ denote the locations where the real system is observed. These locations are not necessarily the same as the design points $\{\BFx_i:i=1,\ldots,k\}$ in general. Nevertheless, we assume for simplicity that $\{\BFt_i:i=1,\ldots,\ell\}\subseteq \{\BFx_i:i=1,\ldots,k\}$. This is a reasonable assumption, since it is usually more expensive to collect real data than to run simulation experiments, and during experiment design of the simulation model we can choose to set the design points to include $\{\BFt_i:i=1,\ldots,\ell\}$. The theory developed in the sequel can be generalized easily to cover the setting where the two sets of locations are arbitrarily different. Further, we assume that $\BFt_i = \BFx_i$ for each $i=1,\ldots,\ell$, since we can change the indexes properly otherwise.

For each $i=1,\ldots, \ell$, let $z_i$ denote the observation of the real system at $\BFx_i$ and $\zeta(\BFx_i)$ denote the corresponding observation error with mean zero, so that
\[z_i=\SFZ(\BFx_i)+\zeta(\BFx_i).\]
For any $\BFx_0$, we want to predict the response $\SFZ(\BFx_0)$ based on both the simulation outputs $\overline{\BFy}=(\overline{y}(\BFx_1),\ldots,\overline{y}(\BFx_k))^\intercal$ and the real data $\BFz\coloneqq(z_1,\ldots,z_\ell)^\intercal$.  Using the \emph{augmented} data set $(\overline{\BFy}, \BFz)$ for prediction represents a key difference between SK-i and SK, since the latter does not account for the model discrepancy and uses only $\overline{\BFy}$ for prediction. The following assumptions are standard in SK literature.

\begin{assumption}\label{asp:normal_errors}
The simulation errors $\{\epsilon_j(\BFx_i):j=1,2,\ldots\}$ are independent normal random variables with mean 0 and variance $\sigma_\epsilon^2(\BFx_i)$, $i=1,\ldots,k$, and $\epsilon_j(\BFx_i)$ is independent of $\epsilon_{j'}(\BFx_{i'})$ if $i\neq i'$. 
\end{assumption}

\begin{assumption}\label{asp:obs_errors}
The observation errors $\{\zeta(\BFx_i):i=1,\ldots,\ell\}$ are independent normal random variables with mean 0 and variance $\sigma_\zeta^2$.
\end{assumption}

\begin{assumption}\label{asp:mutual_ind}
The Gaussian random fields, $\SFM$ and $\SFW$, the simulation errors $\{\epsilon_j(\BFx_i):j=1,\ldots,n_i, i=1,\ldots,k\}$, and the observation errors $\{\zeta(\BFx_i):i=1,\ldots,\ell\}$ are mutually independent.
\end{assumption}

It is easy to see that the augmented data has multivariate normal distribution under Assumptions \ref{asp:normal_errors} and \ref{asp:mutual_ind}; see Proposition \ref{prop:joint_normal} below.  
The following notations are also needed to facilitate the presentation. Let $\BFM(k)\coloneqq (\SFM(\BFx_1),\ldots,\SFM(\BFx_k))^\intercal$, $\BFSigma_{\BFM(k)}$ denote the covariance matrix of $\BFM(k)$,  $\BFSigma_{\BFM(k),\BFM(\ell)}$ denote the covariance matrix between $\BFM(k)$ and $\BFM(\ell)$, and $\BFSigma_{\BFM(k)}(\BFx_0,\cdot)$ denote the $k\times 1$ vector whose $i^{\mathrm{th}}$ component is $\Cov(\SFM(\BFx_0), \SFM(\BFx_i))$, $i=1,\ldots,k$. Moreover, let $\BFSigma_{\BFW}$ denote the covariance matrix of $\BFW\coloneqq (\SFW(\BFx_1),\ldots,\SFW(\BFx_\ell))$, $\BFSigma_\zeta$ denote the covariance matrix of $(\zeta(\BFx_1),\ldots,\zeta(\BFx_\ell))$, and $\BFSigma_{\BFW}(\BFx_0,\cdot)$  denote the $\ell\times 1$ vector whose $i^{\mathrm{th}}$ component is $\Cov(\SFW(\BFx_0), \SFW(\BFx_i))$, $i=1,\ldots,\ell$. Finally, let $\BFF(k)\coloneqq (\BFf(\BFx_1),\ldots,\BFf(\BFx_k))^\intercal$ and $\BFG\coloneqq (\BFg(\BFx_1),\ldots,\BFg(\BFx_\ell))^\intercal$.

\begin{proposition}\label{prop:joint_normal}
Under Assumptions \ref{asp:normal_errors} -- \ref{asp:mutual_ind},
\begin{equation}\label{eq:multi_normal}
\begin{pmatrix}
\overline{\BFy} \\[0.5ex]
 \BFz
\end{pmatrix}
\sim
\mathcal N\left(
\begin{pmatrix}
\BFF(k)\BFbeta \\[0.5ex]
\rho\BFF(\ell)\BFbeta + \BFG\BFgamma
\end{pmatrix},
\BFV
\right),
\end{equation}
where $\BFV$ is a block matrix as follows
\begin{equation}\label{eq:Vmatrix}
\BFV =
\begin{pmatrix}
\BFV_{\BFone\BFone} & \BFV_{\BFone\BFtwo} \\
\BFV_{\BFone\BFtwo}^\intercal & \BFV_{\BFtwo\BFtwo}
\end{pmatrix}
\coloneqq
\begin{pmatrix}
\BFSigma_{\BFM(k)} + \BFSigma_\epsilon &\rho \BFSigma_{\BFM(k),\BFM(\ell)} \\[0.5ex]
\rho \BFSigma_{\BFM(k),\BFM(\ell)}^\intercal & \rho^2\BFSigma_{\BFM(\ell)} + \BFSigma_{\BFW} +\BFSigma_\zeta
\end{pmatrix}.
\end{equation}
\end{proposition}

\begin{theorem}\label{theo:BLUP}
Under Assumptions \ref{asp:normal_errors} -- \ref{asp:mutual_ind},  the BLUP of $\SFZ(\BFx_0)$  that minimizes the MSE is
\begin{equation}\label{eq:BLUP}
\widehat \SFZ(\BFx_0) = \rho\,\BFf^\intercal(\BFx_0) \BFbeta + \BFg^\intercal(\BFx_0)\BFgamma + \BFC^\intercal \BFV^{-1}\left[
\begin{pmatrix}
\overline{\BFy} \\[0.5ex]
\BFz
\end{pmatrix}
 -
\begin{pmatrix}
\BFF(k) & \BFzero \\[0.5ex]
\rho\, \BFF(\ell) & \BFG
\end{pmatrix}
\begin{pmatrix}
\BFbeta \\\BFgamma
\end{pmatrix}
\right],
\end{equation}
where $\BFC$ is a  block vector as follows
\begin{equation*}
\label{eq:Cmatrix}
\BFC =
\begin{pmatrix}
\BFC_{\BFone} \\
\BFC_{\BFtwo}
\end{pmatrix}
\coloneqq
\begin{pmatrix}
\rho \BFSigma_{\BFM(k)}(\BFx_0,\cdot) \\[0.5ex]
\rho^2 \BFSigma_{\BFM(\ell)}(\BFx_0,\cdot)+ \BFSigma_{\BFW}(\BFx_0,\cdot)
\end{pmatrix}.
\end{equation*}
The optimal MSE is
\begin{equation}
\label{eq:optimal_MSE}
\MSE^*\left(\widehat \SFZ(\BFx_0)\right)= \rho^2 \Sigma_{\SFM}(\BFx_0,\BFx_0)  + \Sigma_{\SFW}(\BFx_0, \BFx_0) - \BFC^\intercal \BFV^{-1}\BFC.
\end{equation}
\end{theorem}

The proofs of Proposition \ref{prop:joint_normal} and Theorem \ref{theo:BLUP} can be found in \cite{ZhangZou16} with a slight modification, so we omit the details. 

The expression \eqref{eq:BLUP} can be interpreted as follows. It is easy to show that $\widehat \SFZ(\BFx_0) = \E[\SFZ(\BFx_0)|\overline{\BFy}, \BFz]$, namely, it is the \textit{conditional} expectation of the response given the augmented data. The term $\rho\BFf^\intercal(\BFx_0) \BFbeta + \BFg^\intercal(\BFx_0)\BFgamma$ in \eqref{eq:BLUP} is the \emph{unconditional} expectation of the response, i.e.,  $\E[\SFZ(\BFx_0)]$, which can be seen easily from \eqref{eq:metamodel}. The last summand in \eqref{eq:BLUP}, on the other hand, represents the information from the correlation between the response  and the augmented data. More specifically, $\BFC$ is the covariance vector between $\SFZ(\BFx_0)$ and $(\overline{\BFy}^\intercal, \BFz^\intercal)$, whereas $\BFV$ the covariance matrix of $(\overline{\BFy}^\intercal, \BFz^\intercal)$. Examining the expression of $\BFC$ more closely, we find that $\SFZ(\BFx_0)$ correlates with $\overline{\BFy}$ only through the random field $\SFM$, whereas $\SFZ(\BFx_0)$ correlates with $\overline{\BFy}$ through both $\SFM$ and $\SFW$. This is a consequence of the mutual independence between various random elements in Assumption \ref{asp:mutual_ind}. A similar statement can also be made about the covariance structure of the augmented data. For example, $\Cov(\overline{\BFy}, \BFz) = \rho \BFSigma_{\BFM(k), \BFM(\ell)}$ suggests that $\overline{\BFy}$ and $\BFz$ are correlated because they both involve $\SFM$ in their formulations.

Theorem \ref{theo:BLUP} generalizes a similar result in \cite{kennedy2001bayesian}. In particular, if the simulation model has no simulation errors, i.e., $\BFSigma_\epsilon = \BFzero$, then the BLUP and its MSE in Theorem  \ref{theo:BLUP} are reduced to those in \cite{kennedy2001bayesian} for deterministic simulation models. Theorem \ref{theo:BLUP} also generalizes the counterpart for SK. In particular, by setting $\rho=1$ and removing the observations $\BFz$, we can reduce \eqref{eq:BLUP} and \eqref{eq:optimal_MSE} to \eqref{eq:SK_BLUP} and \eqref{eq:SK_MSE}, respectively.

\section{Analysis of the SK-i Metamodel}\label{sec:SK-i}
In this section, we compare SK-i with two other methods for predicting the response of the real system and demonstrate the advantage of leveraging both the simulation outputs and the observations of the real system jointly for prediction. We also conduct sensitivity analysis and investigate how the MSE of prediction responds to the changes in various aspects of the SK-i metamodel, including the variability in the simulation errors, the variability in the observation errors, the sample size of the simulation outputs, and the sample size of the observations of the real system.

\subsection{Comparison with Other Prediction Methods}
In this section, we compare SK-i with two competing methods for predicting the response of the real system. One method is to apply the metamodel \eqref{eq:metamodel} to the observations of the real system and predict the response of the real system. The predictor is given by Proposition \ref{prop:GPR} below. We refer to this approach as Gaussian process regression (GPR). The proof of Proposition \ref{prop:GPR} is similar to that of Theorem \ref{theo:BLUP}, thereby  deferred to Appendix \ref{app:B}.

\begin{proposition}\label{prop:GPR}
Under Assumptions \ref{asp:normal_errors} and \ref{asp:mutual_ind}, the BLUP of $\SFZ(\BFx_0)$ given $\BFz$ is
\begin{equation}\label{eq:GPR_BLUP}
\widehat \SFZ_{\mathrm{GPR}}(\BFx_0) \coloneqq \rho\,\BFf^\intercal(\BFx_0) \BFbeta + \BFg^\intercal(\BFx_0)\BFgamma  + \BFC_{\BFtwo}^\intercal \BFV_{\BFtwo\BFtwo}^{-1} [\BFz - \rho \BFF(\ell)\BFbeta - \BFG\BFgamma],
\end{equation}
where $\BFC_{\BFtwo} = \rho^2 \BFSigma_{\BFM(\ell)}(\BFx_0,\cdot) + \BFSigma_{\BFW}(\BFx_0, \cdot)$ and $\BFV_{\BFtwo, \BFtwo} = \rho^2\BFSigma_{\BFM(\ell)} + \BFSigma_{\BFW} + \BFSigma_\zeta$. The optimal MSE is
\begin{equation}
\label{eq:GPR_MSE}
\MSE^*\left(\widehat \SFZ_{\mathrm{GPR}}(\BFx_0)\right)
= \rho^2 \Sigma_{\SFM}(\BFx_0,\BFx_0)  + \Sigma_{\SFW}(\BFx_0, \BFx_0)
- \BFC_{\BFtwo}^\intercal \BFV_{\BFtwo\BFtwo}^{-1}\BFC_{\BFtwo}.
\end{equation}
\end{proposition}

A second competing method is to neglect any inadequacy of the simulation model and use SK with the simulation outputs to predict the response of the simulation model as if it were the true response of the real system. The predictor is $\widehat\SFY(\BFx_0)$ given by \eqref{eq:SK_BLUP}. We refer to this method as SK.

We stress here that the three methods use different data for prediction: SK uses only $\overline{\BFy}$, GPR uses only $\BFz$, whereas SK-i uses both. Hence, it is conceivable that SK-i ought to have the most accurate prediction since it uses more data than the other two methods. We show below that this is indeed the case. Specifically, provided that the parameters of these metamodels are known, SK-i has the smallest MSE among the three methods.

\begin{theorem}\label{theo:MSE_comparison}
Let  $\MSE^*_{\mathrm{GPR}}$, $\MSE^*_{\mathrm{SK}}$, and $\MSE^*_{\mathrm{SK-i}}$  denote the MSE for predicting $\SFZ(\BFx_0)$ using $\widehat\SFZ_{\mathrm{GPR}}(\BFx_0)$, $\widehat \SFY(\BFx_0)$, and $\widehat \SFZ(\BFx_0)$, respectively. Then, under Assumptions \ref{asp:normal_errors} -- \ref{asp:mutual_ind},
\begin{enumerate}[label=(\roman*)]
\item
$\MSE^*_{\mathrm{SK-i}} \leq \MSE^*_{\mathrm{GPR}}$, and the equality holds if and only if $\BFC_{\BFone} - \BFV_{\BFone\BFtwo}  \BFV_{\BFtwo\BFtwo}^{-1} \BFC_{\BFtwo} = \BFzero$;
\item
$\MSE^*_{\mathrm{SK-i}} \leq \MSE^*_{\mathrm{SK}}$, and the equality holds if and only if
\[
(\rho-1)\BFf^\intercal(\BFx_0) \BFbeta + \BFg^\intercal(\BFx_0) \BFgamma =0, \quad
(\rho-1) \BFSigma_{\BFM(k)}^\intercal(\BFx_0,\cdot)=\BFzero, \quad\mbox{and}\quad
\BFC_{\BFtwo}-\BFV_{\BFone\BFtwo}^\intercal \BFV_{\BFone\BFone}^{-1} \BFC_{\BFone} =\BFzero.
\]
\end{enumerate}
\end{theorem}

\begin{proof}

It is straightforward to prove $\MSE^*_{\mathrm{SK-i}} \leq \MSE^*_{\mathrm{GPR}}$ and $\MSE^*_{\mathrm{SK-i}} \leq \MSE^*_{\mathrm{SK}}$ by noticing that the three predictors are all linear predictors of the form $a+\BFb^\intercal \overline{\BFy} + \BFc^\intercal \BFz$ for some constant $a$ and vectors $\BFb$ and $\BFc$. The value of $(a, \BFb,\BFc)$ for the SK-i approach is the one that minimizes the MSE of such linear predictors.

The conditions for the equalities, however, rely on explicit calculation.  By \eqref{eq:optimal_MSE} and \eqref{eq:GPR_MSE},
\[
\MSE^*_{\mathrm{GPR}} - \MSE^*_{\mathrm{SK-i}} =  \BFC^\intercal \BFV^{-1} \BFC - \BFC_{\BFtwo}^\intercal \BFV_{\BFtwo\BFtwo}^{-1}\BFC_{\BFtwo}.
\]
Let $\BFS \coloneqq \BFV_{\BFone\BFone} - \BFV_{\BFone\BFtwo}\BFV_{\BFtwo\BFtwo}^{-1}\BFV_{\BFone\BFtwo}^\intercal$ be the Schur complement of $\BFV_{\BFtwo\BFtwo}$. Then,
\[\BFV^{-1} =
\begin{pmatrix}
\BFS^{-1} & - \BFS^{-1} \BFV_{\BFone\BFtwo}\BFV_{\BFtwo\BFtwo}^{-1} \\
- \BFV_{\BFtwo\BFtwo}^{-1} \BFV_{\BFone\BFtwo}^\intercal \BFS^{-1} &   \BFV_{\BFtwo\BFtwo}^{-1} + \BFV_{\BFtwo\BFtwo}^{-1} \BFV_{\BFone\BFtwo}^\intercal  \BFS^{-1} \BFV_{\BFone\BFtwo} \BFV_{\BFtwo\BFtwo}^{-1}
\end{pmatrix}
;
\]
see, e.g., \citet[\S0.8.5]{HornJohnson12}. Hence, by straightforward calculation,
\[\MSE^*_{\mathrm{GPR}} - \MSE^*_{\mathrm{SK-i}} =
[\BFC_{\BFone} - \BFV_{\BFone\BFtwo}  \BFV_{\BFtwo\BFtwo}^{-1} \BFC_{\BFtwo}]^\intercal
\BFS^{-1}
[\BFC_{\BFone} - \BFV_{\BFone\BFtwo}  \BFV_{\BFtwo\BFtwo}^{-1} \BFC_{\BFtwo}].\]
Since $\BFV$ is a non-singular covariance matrix, $\BFV$ is positive definite. Then, $\BFS$ is positive definite by Theorem 7.7.7 of \cite{HornJohnson12}, and thus $\BFS^{-1}$ is positive definite. Hence, $\MSE^*_{\mathrm{SK-i}} = \MSE^*_{\mathrm{GPR}}$ if and only if $\BFC_{\BFone} - \BFV_{\BFone\BFtwo}  \BFV_{\BFtwo\BFtwo}^{-1} \BFC_{\BFtwo} = \BFzero$.

The proof of part (ii) is similar. By \eqref{eq:SK_BLUP} and \eqref{eq:metamodel},
\[\SFZ(\BFx_0) - \widehat\SFY(\BFx_0) = (\rho-1)\BFf^\intercal(\BFx_0) \BFbeta + \BFg^\intercal(\BFx_0) \BFgamma +   \rho\SFM(\BFx_0) + \SFW(\BFx_0)  - \BFSigma^\intercal_{\BFM(k)}(\BFx_0,\cdot)  [\BFSigma_{\BFM(k)}+\BFSigma_\epsilon]^{-1}( \overline{\BFy}-\BFF\BFbeta). \]
Hence,
\begin{equation}\label{eq:bias}
\E[ \SFZ(\BFx_0) - \widehat\SFY(\BFx_0) ] = (\rho-1)\BFf^\intercal(\BFx_0) \BFbeta + \BFg^\intercal(\BFx_0) \BFgamma,
\end{equation}
by Proposition \ref{prop:joint_normal}, and
\begin{align*}
\Var[\SFZ(\BFx_0) - \widehat\SFY(\BFx_0)] =&  \Var[\rho\SFM(\BFx_0) + \SFW(\BFx_0) -  \BFSigma^\intercal_{\BFM(k)}(\BFx_0,\cdot) [\BFSigma_{\BFM(k)}+\BFSigma_\epsilon]^{-1} \overline{\BFy}] \nonumber \\
=& \rho^2\Sigma_{\SFM}(\BFx_0,\BFx_0) + \Sigma_{\SFW}(\BFx_0,\BFx_0) + (1-2\rho)  \BFSigma^\intercal_{\BFM(k)}(\BFx_0,\cdot) [\BFSigma_{\BFM(k)}+\BFSigma_\epsilon]^{-1}\BFSigma_{\BFM(k)}(\BFx_0,\cdot) \nonumber \\
=& \rho^2\Sigma_{\SFM}(\BFx_0,\BFx_0) + \Sigma_{\SFW}(\BFx_0,\BFx_0) + (1-2\rho) \rho^{-2} \BFC_{\BFone}^\intercal \BFV_{\BFone\BFone}^{-1}\BFC_{\BFone},\label{eq:var_SK}
\end{align*}
by direct calculation. Since
\[\MSE^*_{\mathrm{SK}} =  \E[(\SFZ(\BFx_0) - \widehat\SFY(\BFx_0))^2] = [\E[ \SFZ(\BFx_0) - \widehat\SFY(\BFx_0) ]]^2 + \Var[\SFZ(\BFx_0) - \widehat\SFY(\BFx_0)], \]
it follows from \eqref{eq:optimal_MSE} that
\begin{equation}\label{eq:MSE_diff}
\MSE^*_{\mathrm{SK}} - \MSE^*_{\mathrm{SK-i}} = [(\rho-1)\BFf^\intercal(\BFx_0) \BFbeta + \BFg^\intercal(\BFx_0) \BFgamma]^2
+ (1-2\rho) \rho^{-2} \BFC_{\BFone}^\intercal \BFV_{\BFone\BFone}^{-1}\BFC_{\BFone}
+ \BFC^\intercal \BFV^{-1} \BFC.
\end{equation}
Let $\BFT\coloneqq \BFV_{\BFtwo\BFtwo}-\BFV_{\BFone\BFtwo}^\intercal \BFV_{\BFone\BFone}^{-1}\BFV_{\BFone\BFtwo}$ be the Schur complement of $\BFV_{\BFone\BFone}$. Then,
\[\BFV^{-1} =
\begin{pmatrix}
\BFV_{\BFone\BFone}^{-1} + \BFV_{\BFone\BFone}^{-1} \BFV_{\BFone\BFtwo} \BFT^{-1} \BFV_{\BFone\BFtwo}^\intercal \BFV_{\BFone\BFone}^{-1} & -\BFV_{\BFone\BFone}^{-1} \BFV_{\BFone\BFtwo} \BFT^{-1} \\
-\BFT^{-1}\BFV_{\BFone\BFtwo}^\intercal \BFV_{\BFone\BFone}^{-1} & \BFT^{-1}
\end{pmatrix},
\]
and it is easy to show that
\begin{align*}
\MSE^*_{\mathrm{SK}} - \MSE^*_{\mathrm{SK-i}} =&  [(\rho-1)\BFf^\intercal(\BFx_0) \BFbeta + \BFg^\intercal(\BFx_0) \BFgamma]^2 \\
& +   (\rho-1)^2 \rho^{-2} \BFC_{\BFone}^\intercal \BFV_{\BFone\BFone}^{-1}\BFC_{\BFone}
+ [\BFC_{\BFtwo} - \BFV_{\BFone\BFtwo}^\intercal \BFV_{\BFone\BFone}^{-1} \BFC_{\BFone}]^\intercal
\BFT^{-1}
[\BFC_{\BFtwo} - \BFV_{\BFone\BFtwo}^\intercal \BFV_{\BFone\BFone}^{-1} \BFC_{\BFone}].
\end{align*}
Since $\BFV_{\BFone\BFone}^{-1}$ and $\BFT^{-1}$ are both positive definite, $\MSE^*_{\mathrm{SK}} - \MSE^*_{\mathrm{SK-i}}=0$ if and only if all the three summands are zeros, which is equivalent to $(\rho-1)\BFf^\intercal(\BFx_0) \BFbeta + \BFg^\intercal(\BFx_0) \BFgamma=0$, $\BFC_{\BFtwo} - \BFV_{\BFone\BFtwo}^\intercal \BFV_{\BFone\BFone}^{-1} \BFC_{\BFone} = \BFzero$,  and
$(\rho-1)\rho^{-1} \BFC_{\BFone} =\BFzero$. Noticing that $ \BFC_{\BFone} = \rho \BFSigma_{\BFM(k)}(\BFx_0,\cdot) $ completes the proof. 
\end{proof}

We show in Proposition \ref{prop:GPR} that $\SFZ_{\mathrm{GPR}}(\BFx_0)$ is an unbiased predictor  of $\SFZ(\BFx_0)$. Thus,  the difference in MSE between the SK-i method and the GPR method reflects their difference in prediction variance. In particular, compared to using only the observations of the real system for prediction, adding simulation outputs introduces no prediction bias and reduces prediction variance.

By contrast, $\widehat \SFY(\BFx_0)$ is biased for predicting $\SFZ(\BFx_0)$ in general  by \eqref{eq:bias}. The bias stems from  the model discrepancy being discarded by the SK method. Equation \eqref{eq:MSE_diff} further reveals that compared to the SK method, the SK-i method can both eliminate the prediction bias and reduce the prediction variance by taking advantage of observations of the real system.  (Even one observation suffices!)

\subsection{Effects of Simulation and Observation Errors}

We investigate the effects of the variability of the simulation/observation errors on the MSE of the SK-i metamodel. Notice that $\sigma_\epsilon^2(\BFx_i)$, $i=1,\ldots,k$, and $\sigma_\zeta^2$ appear only in the diagonal elements of the matrix $\BFV$ in \eqref{eq:optimal_MSE}, the expression of $\MSE^*(\widehat\SFZ(\BFx_0))$. Let $V_{ii}$ denote the $i^{\mathrm{th}}$ diagonal of $\BFV$, $i=1,\ldots,k+\ell$. Then, standard results of matrix calculus imply that
\[\frac{\partial \MSE^*(\widehat\SFZ(\BFx_0))}{\partial V_{ii}} = \frac{\partial (\BFC^\intercal\BFV^{-1}\BFC)}{\partial V_{ii}} = \BFC^\intercal\left[\BFV^{-1} \frac{\partial \BFV}{\partial V_{ii}}\BFV^{-1}\right]\BFC = (\BFV^{-1}\BFC)_i^2 \geq 0,
\]
where $(\BFV^{-1}\BFC)_i$ is the $i^{\mathrm{th}}$ element of $\BFV^{-1}\BFC$, since $\frac{\partial \BFV}{\partial V_{ii}}$ is a matrix of all zeros except the $i^{\mathrm{th}}$ diagonal element being 1. Hence, $\MSE^*(\widehat\SFZ(\BFx_0))$ is non-decreasing in $V_{ii}$. Notice that
\[
V_{ii} = \left\{
\begin{array}{ll}
\Sigma_{\SFM}(\BFx_i,\BFx_i) + \sigma_\epsilon^2(\BFx_i) / n_i, & \quad i=1,\ldots,k, \\[0.5ex]
\rho^2 \Sigma_{\SFM}(\BFx_{i-k},\BFx_{i-k}) + \Sigma_{\SFW}(\BFx_{i-k},\BFx_{i-k}) + \sigma_\zeta^2, & \quad i = k+1,\ldots,k+\ell.
\end{array}
\right.
\]
 It then follows that $\MSE^*(\widehat\SFZ(\BFx_0))$ is non-decreasing in both $\sigma_\epsilon^2(\BFx_i)$ and $\sigma_\zeta^2$. Indeed, this is an intuitive result. The less variability there is in the simulation errors or in the observation errors, the more accurate it is to predict the response of the real system. It is also easy to see that the MSE can be reduced by increasing the number of simulation replications.

\subsection{Effects of Data Sizes}

Intuitively, the prediction ought to be more accurate by increasing the number of design points (and run simulation models on them), or by increasing the number of observations of the real system. We now establish this result for SK-i formally.

Suppose that the number of design points is increased to $\tilde k>k$ while keeping everything else the same. Let $\overline{\BFy}_{1:k}$ and $\overline{\BFy}_{k+1:\tilde k}$ denote the original simulation outputs and the simulation outputs at the new design points, respectively. Given the data $(\overline{\BFy}_{1:k}, \overline{\BFy}_{k+1:\tilde k}, \BFz)$, we can write the SK-i predictor in the following linear form
\[
a +  \BFb^\intercal \overline{\BFy}_{1:k} +  \tilde\BFb^\intercal \overline{\BFy}_{k+1:\tilde k} +  \BFc^\intercal \BFz,
\]
for some constant $ a$ and some vectors $ \BFb$, $\tilde \BFb$, and $ \BFc$. On other hand, the SK-i predictor given the data  $(\overline{\BFy}_{1:k}, \BFz)$ can also be written in the above linear form with $\tilde\BFb=\BFzero$.  Since the value of $( a,  \BFb, \tilde\BFb,  \BFc)$ for the SK-i predictor given $(\overline{\BFy}_{1:k}, \overline{\BFy}_{k+1:\tilde k}, \BFz)$ is the one that minimizes the MSE among all such linear predictors. Hence, the MSE is non-increasing in $k$, the number of design points. In the same vein, we can show that  the MSE is non-increasing in $\ell$,  the number of observations of the real system.

\section{Parameter Estimation}\label{sec:estimation}

When deriving the BLUP in equation (\ref{eq:BLUP}), we have implicitly assume that the parameters including $(\BFbeta,\BFgamma)$ and those for defining the covariance matrices are given. However, they are generally unknown in practice. We now discuss the parameter estimation for the SK-i metamodel.

\subsection{Maximum Likelihood Estimation}\label{sec:MLE}

We are interested in the maximum likelihood estimation (MLE) and impose the following assumption to make the MLE more tractable.

\begin{assumption}\label{asp:stationary}
The Gaussian random fields $\SFM$ and $\SFW$ are both second-order stationary, namely,
\[\Sigma_{\SFM}(\BFx, \BFx') = \tau_{\SFM}^2\, \CalR_{\SFM}(\BFx-\BFx';\BFtheta_{\SFM})\quad\mbox{and}\quad \Sigma_{\SFW}(\BFx, \BFx') = \tau_{\SFW}^2\, \CalR_{\SFW}(\BFx-\BFx';\BFtheta_{\SFW}),\]
where $\tau_{\SFM}^2$  (resp., $\tau_{\SFW}^2$) is the spatial variance of $\SFM$  (resp., $\SFW$), and $\CalR_{\SFM}$ (resp., $\CalR_{\SFW}$) is the correlation depending only on $\BFx-\BFx'$ and may be a function of some unknown parameters $\BFtheta_{\SFM}$  (resp., $\BFtheta_W$). Moreover,  $\CalR_{\SFM}(\BFzero;\BFtheta_{\SFM}) = \CalR_{\SFW}(\BFzero;\BFtheta_{\SFW}) = 1$ and
\[\lim_{\|\BFx-\BFx'\|\to\infty} \CalR_{\SFM}(\BFx-\BFx';\BFtheta_{\SFM}) = \lim_{\|\BFx-\BFx'\|\to\infty} \CalR_{\SFW}(\BFx-\BFx';\BFtheta_{\SFW}) = 0, \]
where $\|\cdot \|$ denotes the Euclidean norm.
\end{assumption}

Let $\BFR_{\BFM(\ell)}(\BFtheta_{\SFM})$  denote the correlation matrix of $\BFM(\ell)$, $\BFR_{\BFW}(\BFtheta_{\SFW})$ denote the correlation matrix of $\BFW$, $\BFR_{\BFM(k),\BFM(\ell)}$ denote the correlation matrix between $\BFM(k)$ and $\BFM(\ell)$, and $\BFI_\ell$ denote the $\ell\times \ell$ identity matrix.  Then, the covariance matrix $\BFV$ can be expressed as
\[\BFV =
\begin{pmatrix}
\tau^2_{\SFM}\BFR_{\BFM(k)}(\BFtheta_{\SFM}) + \BFSigma_\epsilon & \rho\tau_{\SFM}^2\BFR_{\BFM(k),\BFM(\ell)}(\BFtheta_{\SFM}) \\[0.5ex]
\rho\tau_{\SFM}^2\BFR_{\BFM(k),\BFM(\ell)}^\intercal (\BFtheta_{\SFM}) & \rho^2\tau_{\SFM}^2\BFR_{\BFM(\ell)}(\BFtheta_{\SFM}) + \tau_{\SFW}^2\BFR_{\BFW}(\BFtheta_{\SFW})+\sigma^2_\zeta\BFI_\ell
\end{pmatrix}
\]

Notice that $\BFSigma_\epsilon = \mathrm{Diag}(\sigma_\epsilon^2(\BFx_1)/n_1,\ldots, \sigma_\epsilon^2(\BFx_k)/n_k)$ by Assumption \ref{asp:normal_errors}. For each $ i=1,\ldots,k$,  to estimate $\sigma_\epsilon^2(\BFx_i)$ we use the sample variance of the simulation replications at $\BFx_i$, i.e.,
\begin{equation}\label{eq:sample_var}
\widehat {\sigma}^2_\epsilon(\BFx_i) = \frac{1}{n_i-1}\sum_{j=1}^{n_i}(y_j(\BFx_i)-\overline{y}(\BFx_i))^2.
\end{equation}
It turns out that estimating $\BFSigma_\epsilon$ in this way and plugging the estimate $\widehat{\BFSigma}_\epsilon$ in the BLUP \eqref{eq:SK_BLUP} introduces no prediction bias, which generalizes a similar result for SK in \cite{ankenman2010stochastic}.  % Since we do not assume repeated observations of the real system at the same location, sample variance estimation of $\sigma_\zeta^2$ may be unavailable. Instead, we estimate it using MLE along with other unknown parameters.

\begin{theorem}\label{theo:unbias}
Let  $\widehat\BFSigma_\epsilon = \mathrm{Diag}(\widehat\sigma_\epsilon^2(\BFx_1)/n_1,\ldots, \widehat\sigma_\epsilon^2(\BFx_k)/n_k)$ and
\[\widehat{\widehat{\SFZ}} (\BFx_0)
= \rho\,\BFf^\intercal(\BFx_0) \BFbeta + \BFg^\intercal(\BFx_0)\BFgamma + \BFC^\intercal \widehat\BFV^{-1}\left[
\begin{pmatrix}
\overline{\BFy} \\[0.5ex]
\BFz
\end{pmatrix}
 -
\begin{pmatrix}
\BFF(k) & \BFzero \\[0.5ex]
\rho\, \BFF(\ell) & \BFG
\end{pmatrix}
\begin{pmatrix}
\BFbeta \\\BFgamma
\end{pmatrix}
\right],
 \]
 where $\widehat\BFV$ is the matrix obtained by replacing $\BFSigma_\epsilon$ by $\widehat\BFSigma_\epsilon $ in \eqref{eq:Vmatrix}. Then, under Assumptions \ref{asp:normal_errors} -- \ref{asp:mutual_ind},
 \[\E[\widehat{\widehat{\SFZ}} (\BFx_0)-\SFZ(\BFx_0)] = 0.\]
\end{theorem}

\begin{proof}

It is well known that the sample variance of a set of i.i.d. normal variables is independent of their sample mean; see, e.g., Example 5.6a in \citet[Chapter 5]{RencherSchaalje08}. Hence, $\widehat {\sigma}^2_\epsilon(\BFx_i)$ is independent of $\overline{y}(\BFx_i)$ by Assumptions \ref{asp:normal_errors} and \ref{asp:mutual_ind}. This implies that $\widehat{\BFSigma}_\epsilon$ is independent of $\overline{\BFy}$, and thus  $\widehat{\BFV} $ is independent of $(\overline{\BFy}^\intercal, \BFz)$. It follows that
\begin{align*}
\E\Big[\widehat{\widehat{\SFZ}}(\BFx_0) \Big] = & \E\Big[ \E\Big[\widehat{\widehat{\SFZ}}(\BFx_0)\Big|\widehat\BFV \Big]\Big] \\
=& \E\left[\rho\,\BFf^\intercal(\BFx_0) \BFbeta + \BFg^\intercal(\BFx_0)\BFgamma  +
\BFC^\intercal \widehat\BFV^{-1} \E  \left[
\begin{pmatrix}
\overline{\BFy} \\[0.5ex]
\BFz
\end{pmatrix}
 -
\begin{pmatrix}
\BFF(k) & \BFzero \\[0.5ex]
\rho\, \BFF(\ell) & \BFG
\end{pmatrix}
\begin{pmatrix}
\BFbeta \\\BFgamma
\end{pmatrix}
\right]
\right] \\
=& \rho\,\BFf^\intercal(\BFx_0) \BFbeta + \BFg^\intercal(\BFx_0)\BFgamma,
\end{align*}
where the last equality follows from Proposition \ref{prop:joint_normal}. Hence,   $\E[\widehat{\widehat{\SFZ}}(\BFx_0) ]  = \E[\SFZ(\BFx_0)]  $.  
\end{proof}

Let $\Xi\coloneqq (\rho, \BFbeta,\BFgamma,\tau_{\SFM}^2,\tau_{\SFW}^2,\BFtheta_{\SFM},\BFtheta_{\SFW}, \sigma_\zeta^2)$ denote the collection of the unknown parameters. We write $\BFV(\Xi)$ to stress its dependence on $\Xi$. It follows from Proposition \ref{prop:joint_normal} that, given $\BFSigma_\epsilon$ the log-likelihood function of the augmented data is
\[
\CalL(\Xi) \coloneqq -\frac{1}{2}(k+\ell)\ln(2\pi) - \frac{1}{2}\ln|\BFV(\Xi)| - \frac{1}{2}
\begin{pmatrix}
\overline{\BFy} - \BFF(k)\BFbeta \\[0.5ex]
\BFz - \rho\BFF(\ell)\BFbeta - \BFG\BFgamma
\end{pmatrix}^\intercal
\BFV(\Xi)^{-1}
\begin{pmatrix}
\overline{\BFy} - \BFF(k)\BFbeta \\[0.5ex]
\BFz - \rho\BFF(\ell)\BFbeta - \BFG\BFgamma
\end{pmatrix},
\]
where $|\BFV(\Xi)|$ is the determinant of $\BFV(\Xi)$.

The MLE can be solved via the first-order optimality conditions:  we set the first-order partial derivative of $\CalL(\Xi)$ with respect to each component of $\Xi$ to zero, and solve the resulting system of equations. The derivatives can be calculated using standard results for matrix calculus. We present them in Appendix \ref{app:MLE} and refer to \citet[Chapter 5]{fang2006design} for related numerical methods.

To summarize, given the augmented data $\{y_j(\BFx_i):j=1,\ldots,n_i, i=1,\ldots,k\}$ and $\{z_i:i=1,\ldots,\ell\}$, an SK-i metamodel is constructed as follows.
\begin{enumerate}[label=(\roman*)]
\item
Estimate $\BFSigma_\epsilon$ using $\widehat\BFSigma_\epsilon = \mathrm{Diag}(\widehat{\sigma}_\epsilon^2(\BFx_1),\ldots, \widehat{\sigma}_\epsilon^2(\BFx_k))$, where $\widehat{\sigma}_\epsilon^2(\BFx_i)$ is given by \eqref{eq:sample_var}.
\item
Using  $\widehat\BFSigma_\epsilon$ instead of $\BFSigma_\epsilon$, maximize $\CalL(\Xi)$ to find $\widehat\Xi = (\widehat\rho, \widehat\BFbeta,\widehat\BFgamma,\widehat\tau_{\SFM}^2,\widehat\tau_{\SFW}^2,\widehat\BFtheta_{\SFM},\widehat\BFtheta_{\SFW}, \widehat\sigma_\zeta^2)$.
\item
Predict $\SFZ(\BFx_0)$ by the plug-in predictor
\[\widehat{\widehat{\SFZ}}(\BFx_0) = \widehat\rho\,\BFf^\intercal(\BFx_0) \widehat\BFbeta + \BFg^\intercal(\BFx_0)\widehat\BFgamma + \widehat\BFC^\intercal \widehat\BFV^{-1}\left[
\begin{pmatrix}
\overline{\BFy} \\[0.5ex]
\BFz
\end{pmatrix}
 -
\widehat\BFH
\begin{pmatrix}
\widehat\BFbeta \\ \widehat\BFgamma
\end{pmatrix}
\right],\]
where
\[
\widehat\BFH =
\begin{pmatrix}
\BFF(k) & \BFzero \\[0.5ex]
\widehat\rho\, \BFF(\ell) & \BFG
\end{pmatrix}
,\quad \widehat\BFC =
\begin{pmatrix}
\widehat\rho\,\widehat\tau^2_{\SFM} \BFR_{\BFM(k)}(\BFx_0,\cdot; \widehat\BFtheta_{\SFM}) \\[0.5ex]
\widehat\rho^2 \widehat\tau^2_{\SFM} \BFR_{\BFM(\ell)}(\BFx_0,\cdot;\widehat\BFtheta_{\SFM})+ \widehat\tau^2_{\SFW}\BFR_{\BFW}(\BFx_0,\cdot;\widehat\BFtheta_{\SFW})
\end{pmatrix},
\]
and
\[\widehat \BFV = \begin{pmatrix}
\widehat \tau^2_{\SFM}\BFR_{\BFM(k)}(\widehat\BFtheta_{\SFM}) + \widehat\BFSigma_\epsilon & \widehat\rho\,\widehat\tau_{\SFM}^2\BFR_{\BFM(k),\BFM(\ell)}(\widehat\BFtheta_{\SFM}) \\[0.5ex]
\widehat\rho\,\widehat\tau_{\SFM}^2\BFR_{\BFM(k),\BFM(\ell)}^\intercal (\widehat\BFtheta_{\SFM}) & \widehat\rho^2\widehat\tau_{\SFM}^2\BFR_{\BFM(\ell)}(\widehat\BFtheta_{\SFM}) + \widehat\tau_{\SFW}^2\BFR_{\BFW}(\widehat\BFtheta_{\SFW})+\widehat\sigma^2_\zeta\BFI_\ell
\end{pmatrix}.\]
The MSE estimator is
\[\widehat\rho^2\widehat\tau_{\SFM}^2+ \widehat\tau_{\SFW}^2 - \widehat\BFC^\intercal \widehat\BFV^{-1} \widehat\BFC  + \widehat\BFd^\intercal \left[ \widehat\BFH\widehat\BFV^{-1}\widehat\BFH \right]^{-1} \widehat\BFd ,\]
where $\widehat\BFd =
(\widehat\rho\, \BFf^\intercal(\BFx_0), \BFg^\intercal(\BFx_0))^\intercal
-\widehat\BFH \widehat\BFV^{-1}\widehat\BFC$; see \citet[\S1.5]{stein1999interpolation} for a similar derivation.
\end{enumerate}

\subsection{Experiment Design}

In this section, we discuss briefly how to allocate a total sampling budget of $N$ replications across a set of fixed design points $\{\BFx_1,\ldots,\BFx_k\}$ in order to minimize the integrated MSE (IMSE). Let $\BFn\coloneqq (n_1,\ldots,n_k)^\intercal$ and  $\mathfrak X$ is the experiment design space in $\Real^d$ of interest. Our goal here is to
\begin{equation}\label{allocation_problem}
\begin{aligned}
\mbox{minimize}&\quad  \mathrm{IMSE}(\BFn)\coloneqq \displaystyle\int_{\mathfrak X} \MSE^*(\BFx_0;\BFn)\,\ud\BFx_0 \\
\mbox{s.t.}&\quad  \mathbf{n}^\intercal \BFone_{k}\leq N \\
& \quad n_i\in\mathbb{N}, \;i=1,\ldots,k,
\end{aligned}
\end{equation}
where $ \BFone_{k}$ denotes the $k$-dimensional vector of all ones and  $\MSE^*(\BFx_0;\BFn)$ is given by \eqref{eq:optimal_MSE} and rewritten as follows to emphasize its dependence on $\BFx_0$ and $\BFn$,
\begin{align*}
\MSE^*(\BFx_0;\BFn) = & \rho^2 \Sigma_{\SFM}(\BFx_0,\BFx_0)  + \Sigma_{\SFW}(\BFx_0, \BFx_0) - \BFC^\intercal(\BFx_0) \BFV^{-1}(\BFn)\BFC(\BFx_0) \\
=& \rho^2 \tau_{\SFM}^2 +\tau_{\SFW}^2 -\sum_{i,j=1}^{k+\ell}[\BFV^{-1}(\BFn)]_{ij}C_i(\mathbf{x}_0)C_j(\mathbf{x}_0)
\end{align*}
under Assumption \ref{asp:stationary}, where $C_i$ is the $i^{\mathrm{th}}$ element of $\BFC$. Let $\BFG$ be the $(k+\ell)\times (k+\ell)$ matrix with elements $G_{ij} = \int_{\mathfrak X}C_i(\mathbf{x}_0)C_j(\mathbf{x}_0)\ud \BFx_0 $. Then, 
\[
\IMSE(\BFn) = \rho^2 \tau_{\SFM}^2 +\tau_{\SFW}^2 - \BFone_{k+\ell}^\intercal [\BFG\circ \BFV^{-1}(\BFn)] \BFone_{k+\ell},\]
where $\circ$ denotes the element-wise product of matrices. 

To obtain a tractable solution to the optimization problem \eqref{allocation_problem}, we relax its integrality constraint and replace it with $n_i\geq 0$, $i=1,\ldots,k$. Then, we can form the Lagrangian 
\[
\mathscr{L}(\BFn)\coloneqq\rho^2 \tau_{\SFM}^2 +\tau_{\SFW}^2 - \BFone_{k+\ell}^\intercal [\BFG\circ \BFV^{-1}(\BFn)] \BFone_{k+\ell}+\lambda(N-\BFone_{k}^\intercal\BFn).
\]
The first-order optimality conditions are 
\begin{equation}\label{eq:first-order-opt}
\frac{\partial \mathscr{L}(\BFn)}{\partial n_i} =  - \BFone_{k+\ell}^\intercal \left[\BFG\circ \frac{\partial }{\partial n_i}\BFV^{-1}(\BFn)\right] \BFone_{k+\ell}-\lambda = 0,\quad i=1,\ldots,k.
\end{equation}
Notice that for $i=1,\ldots,k$,
\begin{equation}\label{eq:deriv_V_inverse}
\frac{\partial }{\partial n_i}\BFV^{-1}(\BFn) = -\BFV^{-1}(\BFn) \left[ \frac{\partial }{\partial n_i}\BFV(\BFn)\right] \BFV^{-1}(\BFn) = \frac{\sigma^2_\epsilon(\BFx_i)}{n_i^2} [\BFV^{-1}(\BFn) \BFJ^{(ii)} \BFV^{-1}(\BFn)]
\end{equation}
where $\BFJ^{(ii)}$ is a $(k+\ell)\times (k+\ell)$ matrix with 1 in position $(i,i)$ and zeros elsewhere.  It can be shown by direct calculation that 
\begin{equation}\label{eq:hadamard_prod}
\BFone_{k+\ell}^\intercal \left[\BFG\circ  [\BFV^{-1}(\BFn) \BFJ^{(ii)} \BFV^{-1}(\BFn)]\right] \BFone_{k+\ell} = [\BFV^{-1}(\BFn)\BFG \BFV^{-1}(\BFn)]_{ii},
\end{equation}
using the fact that $\BFG$ is a symmetric matrix. It then follows from \eqref{eq:first-order-opt}, \eqref{eq:deriv_V_inverse}, and \eqref{eq:hadamard_prod} that
\[\frac{\sigma^2_\epsilon(\BFx_i)}{n_i^2} [\BFV^{-1}(\BFn)\BFG \BFV^{-1}(\BFn)]_{ii} = \lambda, \quad i=1,\ldots,k. \]
Hence, the optimal solution to \eqref{allocation_problem} satisfies $n_i^*\propto \sqrt{\sigma^2_\epsilon(\BFx_i) [\BFV^{-1}(\BFn)\BFG \BFV^{-1}(\BFn)]_{ii}}$. When $N$ is large enough, we have $\BFSigma_\epsilon \approx \BFzero$ and thus
\[\BFV(\BFn) \approx 
\begin{pmatrix}
\tau^2_{\SFM}\BFR_{\BFM(k)}(\BFtheta_{\SFM}) & \rho\tau_{\SFM}^2\BFR_{\BFM(k),\BFM(\ell)}(\BFtheta_{\SFM}) \\[0.5ex]
\rho\tau_{\SFM}^2\BFR_{\BFM(k),\BFM(\ell)}^\intercal (\BFtheta_{\SFM}) & \rho^2\tau_{\SFM}^2\BFR_{\BFM(\ell)}(\BFtheta_{\SFM}) + \tau_{\SFW}^2\BFR_{\BFW}(\BFtheta_{\SFW})+\sigma^2_\zeta\BFI_\ell
\end{pmatrix}
\coloneqq \widetilde \BFV.
\]
This suggests the following approximate solution to \eqref{allocation_problem}: 
\begin{equation}\label{eq:optimal_allocation}
n_i^*\approx N \frac{\sqrt{\sigma^2_\epsilon(\BFx_i) [\widetilde\BFV^{-1}\BFG \widetilde\BFV^{-1}]_{ii}}}{\sum_{j=1}^k\sqrt{\sigma^2_\epsilon(\BFx_j) [\widetilde\BFV^{-1}\BFG \widetilde\BFV^{-1}]_{jj}}},\quad i=1,\ldots,k.
\end{equation}
In practice, neither $\sigma^2_\epsilon(\BFx_i)$ nor the parameters $\Xi=(\rho, \BFbeta,\BFgamma,\tau_{\SFM}^2,\tau_{\SFW}^2,\BFtheta_{\SFM},\BFtheta_{\SFW}, \sigma_\zeta^2)$ are unknown. We can adopt a two-stage design strategy as follows. In the first stage, we allocate $n_0$ replications to each design point $\BFx_i$, $i=1,\ldots,k$ and estimate $\sigma^2_\epsilon(\BFx_i)$ and $\Xi$ using \eqref{eq:sample_var} and MLE, respectively. (The value of $n_0$ should  exceed 10 to obtain meaningful estimates.) In the second stage, we allocate the $N-km_0$ additional replications among the design points using the approximate formula \eqref{eq:optimal_allocation}, update the estimates of $\BFSigma_\epsilon$ and $\Xi$, and predict $\SFZ(\BFx_0)$ following the procedure described at the end of \S\ref{sec:MLE}.

\section{Numerical Experiment}\label{sec:experiments}
In this section, we compare numerically the three competing methods (i.e., GPR, SK, and SK-i) for in terms of prediction accuracy. The following example is adopted from Simulation Optimization Library (\texttt{simopt.org}). Consider a production line consisting of three service stations, each having a single server and a finite capacity. The parts are processed at each station on a first-in-first-out basis. Once its service is completed at station $n$, a part is moved to station $n+1$, provided that the downstream station is not full; otherwise, the part is blocked, staying at station $n$, and occupies the server. We assume that the parts arrive at the production line following a Poisson process with unit rate and each station has a capacity of 5. We also assume that  the ``real'' system has service times with the gamma distribution whereas the inadequate simulation model has service times with the exponential distribution. 

Define the design variable $\BFx=(x_1,\ldots,x_6)^\intercal$ as follows: $x_i$ and $x_{i+3}$ are respectively the mean and the variance of the service time of station $i$, $i=1,2,3$. Suppose that the performance measure of interest $\SFZ(\BFx)$ is the expected sojourn time through the production line, which is estimated based on all the parts generated within a time horizon $T=60000$.

Let $\mathfrak X\subset \Real^6$ denote the design space of interest. We consider three design spaces
\begin{itemize}
\item 
$\mathfrak X_1=[0.2,0.3]^2\times[0.7,0.95]\times[0.05,0.1]^2\times[0.8,0.9]$; 
\item 
$\mathfrak X_2=[0.2,0.3]^2\times[0.7,0.95]\times[0.2,0.3]^2\times[1.0,1.1]$;
\item 
$\mathfrak X_3=[0.2,0.3]^2\times[0.7,0.95]\times[0.3,0.4]^2\times[1.2,1.3]$.
\end{itemize}
Notice that the variance of the exponential distribution is the square of the mean. The three design spaces basically represent scenarios where the simulation model is close to the real system, moderately inadequate, and highly inadequate, respectively. 

 Suppose that the simulation model is executed at $k=40$ design points in $\mathfrak X$, which are generated using  Latin hypercube sampling (LHS); see \citet[Chapter 2]{fang2006design}. Given a computational budget $N=\sum_{i=1}^kn_i$, the number of replications $n_i$ at each design point $\BFx_i$ is computed via \eqref{eq:optimal_allocation}. Moreover, suppose that the performance measure of the ``real'' system is observed at $\ell=20$ locations, which are chosen randomly from the $k$ design points. The observations $(z_1,\ldots,z_\ell)$ are generated by simulating the ``real'' system with $10000$  total number of replications allocated to the $\ell$ locations following the rule  \eqref{eq:optimal_allocation} as well. 

In order to compare the prediction accuracy, we predict $\SFZ(\BFx_0)$ at $K=1000$ different locations $\{\BFx^{(i)}:i=1,\ldots,K\}$  generated using LHS. Since the predictor $\hat \SFZ(\BFx^{(i)})$ is subject to the randomness of both simulation errors and observations errors, we conduct $R=100$ macro-replications of the experiment and use the following estimated MSE (EMSE) to compare the three prediction methods,
\begin{equation}\label{eq:EMSE}
\mbox{EMSE} = \frac{1}{R}\sum_{r=1}^R\frac{1}{K}\sum_{i=1}^K\left[\hat\SFZ(\BFx^{(i)})-\SFZ(\BFx^{(i)})\right]^2,
\end{equation}
where  the subscript $r$ denotes the $r^{\mathrm{th}}$ macro-replication, and the unknown true value $\SFZ(\BFx)$ is replaced with estimates based on extensive simulation until errors are negligible.

We assume $\BFf(x)\equiv 1$ and $\BFg(x)\equiv 1$ and adopt the squared exponential correlation function 
\[\CalR_{\SFM}(\BFx-\BFx';\theta_{\SFM}) = \exp\left(-\theta_{\SFM}\sum_{i=1}^6(x_i-x'_i)^2\right) \quad\mbox{and}\quad 
\CalR_{\SFW}(\BFx-\BFx';\theta_{\SFW}) = \exp\left(-\theta_{\SFW}\sum_{i=1}^6(x_i-x'_i)^2\right). 
 \]

\begin{table}[t]
\centering    
  \begin{tabular}{ccccccccccccc}
   \toprule
    \multirow{2}{*}{} &&
      \multicolumn{3}{c}{$\mathfrak X_1$} &&
      \multicolumn{3}{c}{$\mathfrak X_2$} &&
      \multicolumn{3}{c}{$\mathfrak X_3$} \\
      \cmidrule{3-5} \cmidrule{7-9} \cmidrule{11-13}
      && {$N=7000$} && {$N=700$} && {$N=7000$} && {$N=700$}&& {$N=7000$} && {$N=700$} \\
      \midrule
    GPR && 0.752 && 0.782  && 1.031 && 1.001 && 2.425 && 2.431 \\
    SK && 0.555 && 0.632 && 5.050 && 5.872 && 17.27 && 18.71\\
     SK-i && 0.461 &&0.573&& 0.611&& 0.703 && 1.643 && 1.726 \\
    \bottomrule
  \end{tabular}
  \caption{EMSE  for the Production Line Example.} \label{tab:production}
\end{table}

The results are presented in Table $\ref{tab:production}$. As suggested by Theorem \ref{theo:MSE_comparison}, SK-i has the best prediction performance in all the cases. Incorporating observations of the real system can improve substantially the prediction accuracy of SK, especially when the simulation model is highly inadequate as in the case of $\mathfrak X_3$. On the other hand, incorporating the simulation outputs can improve substantially the prediction accuracy of GPR as well, even if the simulation model is inadequate. 

% We have several findings from the results presented in Table $\ref{tab:production}$. First, SK-i outperforms both GPR and SK significantly in all the cases. This is consistent with Theorem \ref{theo:MSE_comparison}, which states that SK-i has the smallest MSE among the three. Second, one should be cautioned when using a simulation model to predict the response surface of the real system, because the prediction may be rather sensitive to the model discrepancy. For example, in the cases of $\mathfrak X_1$ (i.e., small model discrepancy), the prediction accuracy of SK is still considerably lower than SK-i. 

% Third, regardless of the magnitude of the model discrepancy (either $\sigma^2=1.1$ or $\sigma^2=2.5$), incorporating a proper simulation model in a prediction procedure can improve the prediction accuracy significantly, relative to solely relying on observations of the real system. This is because the simulation outputs, which essentially represent domain knowledge, can provide additional information on the structure of the unknown response surface. 

% Last but not least, the magnitude of the observation errors has a considerable impact on the prediction accuracy of the SK-i metamodel. Highly noisy observations of the real system lead to poor prediction. For example, consider the case with $\sigma=1.1$, $m=10^6$, and $N=100$, i.e., large magnitude of observation errors and small model discrepancy. The difference in prediction accuracy between SK-i and the others is significantly smaller in this case than the other cases. 

\section{Effect of Common Random Numbers}\label{sec:CRN}

Common Random Numbers (CRN) is a variance reduction technique that is applied widely in practice thanks to its ease of use. It is known that use of CRN generally increases the MSE of the SK metamodel, thereby leading to its deteriorated performance for predicting the response surface; see, e.g., \cite{chen2012effects}. Nevertheless, this is not necessarily true when the model discrepancy is taken into account. The effect of CRN is significantly more complex for the SK-i metamodel. It may be either detrimental or beneficial to prediction, depending on the magnitude of the observation errors. In this section, we first analyze the effect of CRN via a stylized models to gain insights, and then demonstrate it numerically for a general setting. 

Notice that if $\rho=0$, then the simulation model provides no information about the real system by the definition of $\rho$ in \eqref{real_formulation}. This implies that the use of CRN has no effect on the prediction of the real system, regardless of the correlation structure introduced in the simulation errors. This can also be seen easily by setting $\rho=0$ in \eqref{eq:optimal_MSE}, leading to
\begin{align*}
\MSE^*(\hat\SFZ(\BFx_0)) = & \Sigma_{\SFW}(\BFx_0,\BFx_0) - 
\begin{pmatrix}
\BFzero \\ 
\BFSigma_{\BFW}(\BFx_0,\cdot) 
\end{pmatrix}^\intercal 
\begin{pmatrix}
\BFSigma_{\BFM(k)} + \BFSigma_\epsilon & \BFzero \\
\BFzero & \BFSigma_{\BFW} + \BFSigma_\zeta
\end{pmatrix}^{-1}
\begin{pmatrix}
\BFzero \\ 
\BFSigma_{\BFW}(\BFx_0,\cdot) 
\end{pmatrix} \\
= & \Sigma_{\SFW}(\BFx_0,\BFx_0) - \BFSigma_{\BFW}^\intercal (\BFx_0,\cdot) [\BFSigma_{\BFW} + \BFSigma_\zeta]^{-1}\BFSigma_{\BFW}(\BFx_0,\cdot),
\end{align*}
which is independent of $\BFSigma_\epsilon$. Since CRN takes effect only through the simulation outputs, it follows that CRN has no effect on $\MSE^*(\hat\SFZ(\BFx_0))$.  However, the case of $\rho=0$ rarely occurs in practice since the simulation model is constructed to approximate the real system in the first place. Hence, we assume without loss of generality that $\rho\neq 0$ in the sequel.

% To make the analysis tractable, in \S\ref{sec:two_point_model} and \S\ref{sec:k_point_model}

\subsection{A Two-Point Model}\label{sec:two_point_model}
Consider the case of $k=2$ and $\ell=1$, that is, the simulation model is executed at $\BFx_1$ and $\BFx_2$ and the performance of the real system is observed at $\BFx_1$. The use of CRN introduces dependence between the simulation errors at different design points, and thus $\BFSigma_\epsilon$, the covariance matrix of $(\overline{\epsilon}(\BFx_1),\ldots,\overline{\epsilon}(\BFx_k))$, is no longer a diagonal matrix. In particular, the anticipated effect of CRN is to cause its off-diagonal elements to be positive. In order make the analysis tractable, we make the following assumption that is standard in simulation literature for developing insight.

\begin{assumption}\label{asp:CRN} 
The sample average of the simulation errors, $(\overline{\epsilon}(\BFx_1),\overline{\epsilon}(\BFx_2))$, have bivariate normal distribution with mean $\BFzero$ and covariance matrix 
\[\BFSigma_\epsilon = v 
\begin{pmatrix} 
1 & \omega  \\
\omega & 1
\end{pmatrix}
\]
for some $v>0$ and $\omega\in[0,1]$. Moreover, 
\[\Corr(\SFM(\BFx_0),\SFM(\BFx_1))=\Corr(\SFM(\BFx_0),\SFM(\BFx_2))=\Corr(\SFW(\BFx_0),\SFW(\BFx_1))=r_0,\]
for some $r_0>0$.
\end{assumption}

Let $\MSE^*(\omega)$ denote the MSE of the SK-i predictor as a function of $\omega$. Then, we can determine whether the MSE is increasing or decreasing in $\omega$ by analyzing the sign of $\frac{\ud \MSE^*(\omega)}{\ud \omega}$. By \eqref{eq:optimal_MSE}, 
\begin{equation*}\label{eq:long_MSE}
\begin{aligned}
 & \MSE^*(\omega)\\ =& \rho^2\tau_{\SFM}^2+\tau_{\SFW}^2 \\
& - 
\begin{pmatrix}
\rho \tau_{\SFM}^2 
\begin{pmatrix}
r_0 \\ 
r_0 
\end{pmatrix} \\
\rho^2 \tau_{\SFM}^2 r_0 + \tau_{\SFW}^2 r_0
\end{pmatrix}^\intercal
\begin{pmatrix}
\tau_{\SFM}^2  
    \begin{pmatrix}
    1 & r_{12} \\
    r_{12} & 1 
    \end{pmatrix} + 
v 
    \begin{pmatrix}
    1 & \omega \\
    \omega & 1
    \end{pmatrix}
 &  
\rho \tau_{\SFM}^2 
    \begin{pmatrix}
     1 \\ 
     r_{12}
    \end{pmatrix}
\\
\rho \tau_{\SFM}^2 
    \begin{pmatrix}
    1 & &  r_{12}
    \end{pmatrix}
& \rho^2 \tau_{\SFM}^2 + \tau_{\SFW}^2 + \sigma_\zeta^2
\end{pmatrix}^{-1}   
\begin{pmatrix}
\rho \tau_{\SFM}^2 
    \begin{pmatrix}
    r_0 \\ 
    r_0 
    \end{pmatrix} \\
\rho^2 \tau_{\SFM}^2 r_0 + \tau_{\SFW}^2 r_0
\end{pmatrix},
\end{aligned}
\end{equation*}
where $r_{12}=\Corr(\SFM(\BFx_1), \SFM(\BFx_2))$. The analysis of $\frac{\ud \MSE^*(\omega)}{\ud \omega}$ is straightforward but lengthy. We present the result below but defer the explicit calculations to Appendix \ref{app:CRN}. 

\begin{theorem}\label{theo:CRN_2point}
Suppose that Assumptions \ref{asp:obs_errors} --  \ref{asp:CRN} hold in the two-point model. If $\rho\neq 0$, then 
\begin{enumerate}[label=(\roman*)]
\item 
$\MSE^*(\omega)$ is decreasing in $\omega\in[0,1]$ if $\sigma_\zeta^2 \leq \tau_{\SFM}^2\tau_{\SFW}^2 r_{12}(1-r_{12})/[\tau_{\SFM}^2(1-r_{12})+v]$;
\item 
$\MSE^*(\omega)$ is increasing in $\omega\in[0,1]$ if $\sigma_\zeta^2 \geq \tau_{\SFW}^2r_{12}+v(\rho^2+\tau_{\SFW}^2/\tau_{\SFM}^2)$;
\item 
otherwise, $\MSE^*(\omega)$ is increasing in $\omega\in[0,\omega^*]$ and decreasing in $\omega\in[\omega^*,1]$ for some $\omega^*\in(0,1)$.
\end{enumerate}
\end{theorem}

Theorem \ref{theo:CRN_2point} represents a stark contrast to the prior result in simulation literature that the use of CRN generally increases the MSE of SK. The contrast stems from the presence of two distinct response surfaces in our context -- that of the simulation model $\SFY(\BFx)$ and that of the real system $\SFZ(\BFx)$, whereas only the former is of relevance in typical usage of SK. Moreover,  the data used by SK-i consists of two parts, i.e., the simulation outputs  $\overline{\BFy}$ and the observations of the real system $\BFz$.  CRN introduces positive dependence in the errors of the former, but has no effect on the latter. Using CRN is indeed detrimental to the prediction of $\SFY(\BFx)$. It is, however, not necessarily the case for the prediction of $\SFZ(\BFx)$. 

Statement (i) of Theorem \ref{theo:CRN_2point} is of particular interest. It suggests that if the real system is observed with little errors, i.e., $\sigma_\zeta^2\approx 0$, then the use of CRN is beneficial to the prediction of $\SFZ(\BFx)$, and the benefit increases as $\omega$ increases. This can be interpreted intuitively as follows. Following \eqref{real_formulation}, the prediction of $\SFZ(\BFx)$ essentially comprises the prediction of $\SFY(\BFx)$ and the prediction of the model discrepancy $\delta(\BFx)$. The additional positive dependence that CRN introduces in $\overline{\BFy}$  helps SK-i utilize the augmented data $(\overline{\BFy},\BFz)$  more effectively for quantifying $\delta(\BFx)$, making its prediction more accurate. Consequently, the net effect of CRN depends on whether the benefit of the use of CRN in predicting $\delta(\BFx)$ dominates its detriment to the prediction of $\SFY(\BFx)$, or the opposite is true. If $\BFz$ has negligible errors, then the net effect of CRN is beneficial. On the other hand, as suggested by statement (ii), the use of CRN turns detrimental if $\BFz$ has great errors, i.e., $\sigma_\zeta^2$ is sufficiently large. 

If $\sigma_\zeta^2$ is in the middle range of values, then the effect of CRN depends additionally on the value of $\omega$. Since $\MSE^*(\omega)$ first increases and then decreases as $\omega$ increases from 0 to 1, we expect that $\MSE^*(\omega)\geq \MSE^*(0)$ for small $\omega$. The scenario where  $\MSE^*(\omega)<\MSE^*(0)$ is possible but not necessary to occur. 

There are two messages from the analysis of this two-point model. First, in the presence of model discrepancy, the interplay between various types of uncertainty arising from the two distinct response surfaces and the two sets of data is substantially more sophisticated than that in the SK metamodel itself, which involves only $\SFY(\BFx)$ and $\overline{\BFy}$. As a result, the effect of CRN for SK-i is significantly more complex than it is for SK. Second, albeit counterproductive for  predicting $\SFY(\BFx)$, CRN is indeed helpful for predicting $\SFZ(\BFx)$, provided that the errors in $\BFz$ are small enough. Otherwise, CRN is not recommended for SK-i.

\subsection{Illustration}\label{sec:illustration_CRN}
We have shown via a simple example that the effect of CRN on the prediction of SK-i depends in a nontrivial way on the accuracy of the observations of the real system relative to the magnitude of other types of uncertainty in the metamodel. Despite the fact that the example is highly stylized and imposes stringent constraints on the values of the parameters,  the insights developed there are indeed valid in general as illustrated below numerically.  

Let $\SFM(x)$ and $\SFW(x)$ be two independent one-dimensional Gaussian random fields with $x\in[0,1]$. Their covariance functions are $\Cov(\SFM(x),\SFM(x'))=\tau_{\SFM}^2\exp(-\theta_{\SFM}(x-x')^2)$ and $\Cov(\SFW(x),\SFW(x'))=\tau_{\SFW}^2\exp(-\theta_{\SFW}(x-x')^2)$, respectively. Suppose that the response surface of the simulation model $\SFY(x)$ is a random realization of $\SFM(x)$, and that of the real system $\SFZ(x)$ is the sum of $\SFY(x)$ and a random realization of  $\SFW(x)$, i.e., $\SFY(x)=\SFM(x)$ and $\SFZ(x)=\SFY(x)+\SFW(x)$. 

We set the design points to be $\{0.0, 0.1,\ldots,0.9,1.0\}$, so $k=11$. The simulation errors $(\epsilon(x_1),\ldots,\epsilon(x_k))$ at the design points $(x_1,\ldots,x_k)$ are generated from the multivariate normal distribution with mean $\BFzero$, marginal variance $\Var(\epsilon(x_i)) = \sigma_\epsilon^2 $, $i=1,\ldots,k$, and correlation $\Corr(\epsilon(x_i),\epsilon(x_j))=\omega>0$, $i\neq j$. Then, for each $i=1,\ldots,k$, a simulation output at $x_i$ is $\SFY(x_i) + \epsilon(x_i)$ and we make 10 replications. Moreover, we assume that $\SFZ(x)$ is observed at $\{0.0,0.2,\ldots,0.8,1.0\}$, so $\ell=6$. For each $i=1,\ldots,\ell$, the observation at $x_i$ is generated via $z_i=\SFZ(x_i)+\zeta_i$, where $\zeta_i$ has normal distribution with mean 0 and variance $\sigma_\zeta^2$. We vary the value of $\omega\in[0,1]$ and specify the other parameters as follows:  $\tau_{\SFM}=\tau_{\SFW}=1$, $\theta_{\SFM}=\theta_{\SFW}=5,10,20,30$, $\sigma_\epsilon^2=1,10$, and $\sigma_\zeta^2=0.01,0.1,10$. 

Given a specification of the parameters $(\omega,\tau_{\SFM},\tau_{\SFW},\theta_{\SFM},\theta_{\SFW},\sigma_\epsilon^2,\sigma_\zeta^2)$, we construct 100 problem instances, each of which corresponds to a pair of surfaces $(\SFY(x),\SFZ(x))$ that are generated randomly based on $\SFM(x)$ and $\SFW(x)$. For each problem instance, we conduct 100 macro-replications of the following: 
\begin{enumerate}[label=(\roman*)]
\item 
Generate simulation outputs at the $k$ design points using CRN; generate the observations of $\SFZ(x)$ at $x_i$, $i=1,\ldots,\ell$.
\item 
Compute the predictor $\hat\SFZ(x_0)$ each $x_0=\frac{i}{100}$, $i=1,\ldots,100$. 
\end{enumerate}
Then, we compute $\mbox{EMSE}$ by \eqref{eq:EMSE} after the 100 macro-replications. To facilitate the comparison between CRN and independent sampling, we compute 
\begin{equation}\label{eq:EMSE_ratio}
\frac{\mathrm{EMSE}(\omega,\tau_{\SFM},\tau_{\SFW},\theta_{\SFM},\theta_{\SFW},\sigma_\epsilon^2,\sigma_\zeta^2)}{\mbox{EMSE}(0,\tau_{\SFM},\tau_{\SFW},\theta_{\SFM},\theta_{\SFW},\sigma_\epsilon^2,\sigma_\zeta^2)}
\end{equation}
for each problem instance, and then average this ratio over all the problem instances. Obviously,  the ratio is less (resp., greater) than 1, if the use of CRN is beneficial (resp., detrimental) to the prediction. The results are presented in Figure \ref{fig:CRN_1} for $\sigma_\epsilon^2=1$ and Figure \ref{fig:CRN_2} for $\sigma_\epsilon^2=10$. 

Although it is difficult to obtain analytical results similar to Theorem \ref{theo:CRN_2point} for the general case, Figure \ref{fig:CRN_1} and Figure \ref{fig:CRN_2} confirm the validity of the insights developed in \S\ref{sec:two_point_model}. The effect of CRN on the prediction accuracy of SK-i is complex and depends on the interplay between various parameters. In particular, the three kinds of behavior suggested by Theorem \ref{theo:CRN_2point} -- (i) decreasing, (ii) increasing, and (iii) first increasing then decreasing --- are exactly the same three kinds of behavior shown in our numerical experiment. 

First, if the observation errors of the real system are sufficiently small, then the EMSE is an decreasing function of $\omega\in[0,1]$. For example, most of the plots corresponding to $\sigma_\zeta^2=0.01$ in Figure \ref{fig:CRN_1} and Figure \ref{fig:CRN_2}  represent decreasing functions on $[0,1]$. 

Second, if the observation errors of the real system are sufficiently large, then the EMSE is an increasing function of $\omega\in[0,1]$. The plot corresponding to $\sigma_\zeta^2=10$ in the top-left pane of Figure \ref{fig:CRN_2} is an obvious example. 

Third, the behavior suggested by statement (iii) of Theorem \ref{theo:CRN_2point} appears to be very common one in our numerical experiment. For example, the plots corresponding to $\sigma_\zeta^2=0.1$ and $\sigma_\zeta^2=10$ in the top-right pane of Figure \ref{fig:CRN_2} both behave this way. A critical difference between them is that the former drops below 1 eventually whereas the latter remains above 1. This suggests that if $\sigma_\zeta^2$ has a moderate value, then for small values of $\omega$ the use of CRN is detrimental to the prediction, whereas for large values of $\omega$ it can be either detrimental or beneficial depending on other factors.  

At last, a finding that is not an implication of Theorem \ref{theo:CRN_2point} can be made by comparing Figure \ref{fig:CRN_1} against Figure \ref{fig:CRN_2} to gain insight about the role of $\sigma_\epsilon^2$, the variance of the simulation outputs.  In particular, in a situation where the use of CRN is beneficial (e.g., the case of $\sigma^2_\zeta=0.01$ and $\theta_{\SFM}=\theta_{\SFW}=20$),  it can be seen that the ratio \eqref{eq:EMSE_ratio} takes a larger value for $\sigma_\epsilon^2=1$ in Figure \ref{fig:CRN_1} than for $\sigma_\epsilon^2=10$ in Figure \ref{fig:CRN_2}. This suggests that the beneficial effect of CRN, if there is any, is amplified by the variance of the simulation errors. 

\begin{figure}[t]
\begin{center}
\includegraphics[width=0.9\textwidth]{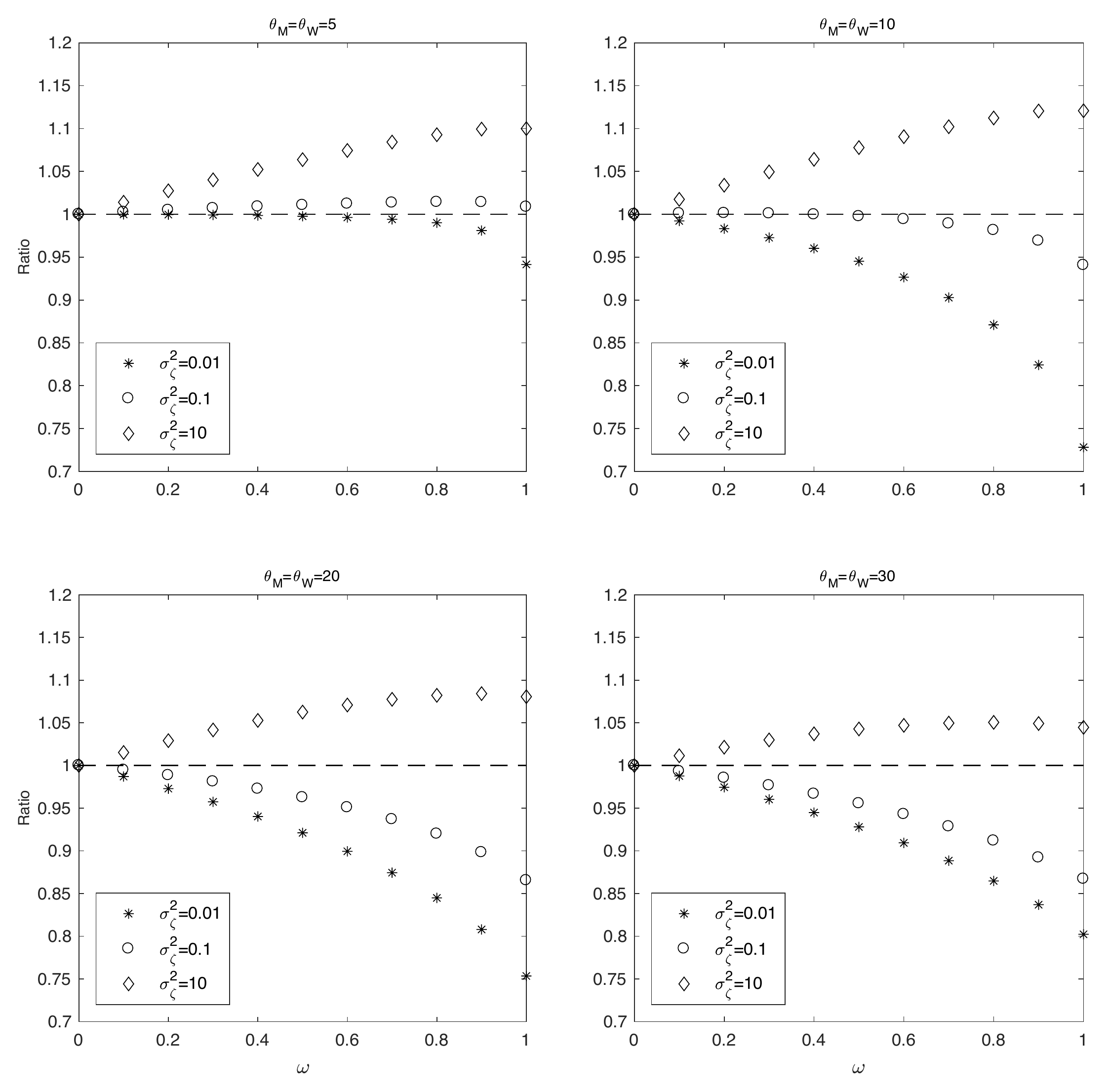}
\caption{The EMSE Ratio \eqref{eq:EMSE_ratio} for $\sigma_\epsilon^2=1$. } \label{fig:CRN_1}
\end{center}
\end{figure}

\begin{figure}[t]
\begin{center}
\includegraphics[width=0.9\textwidth]{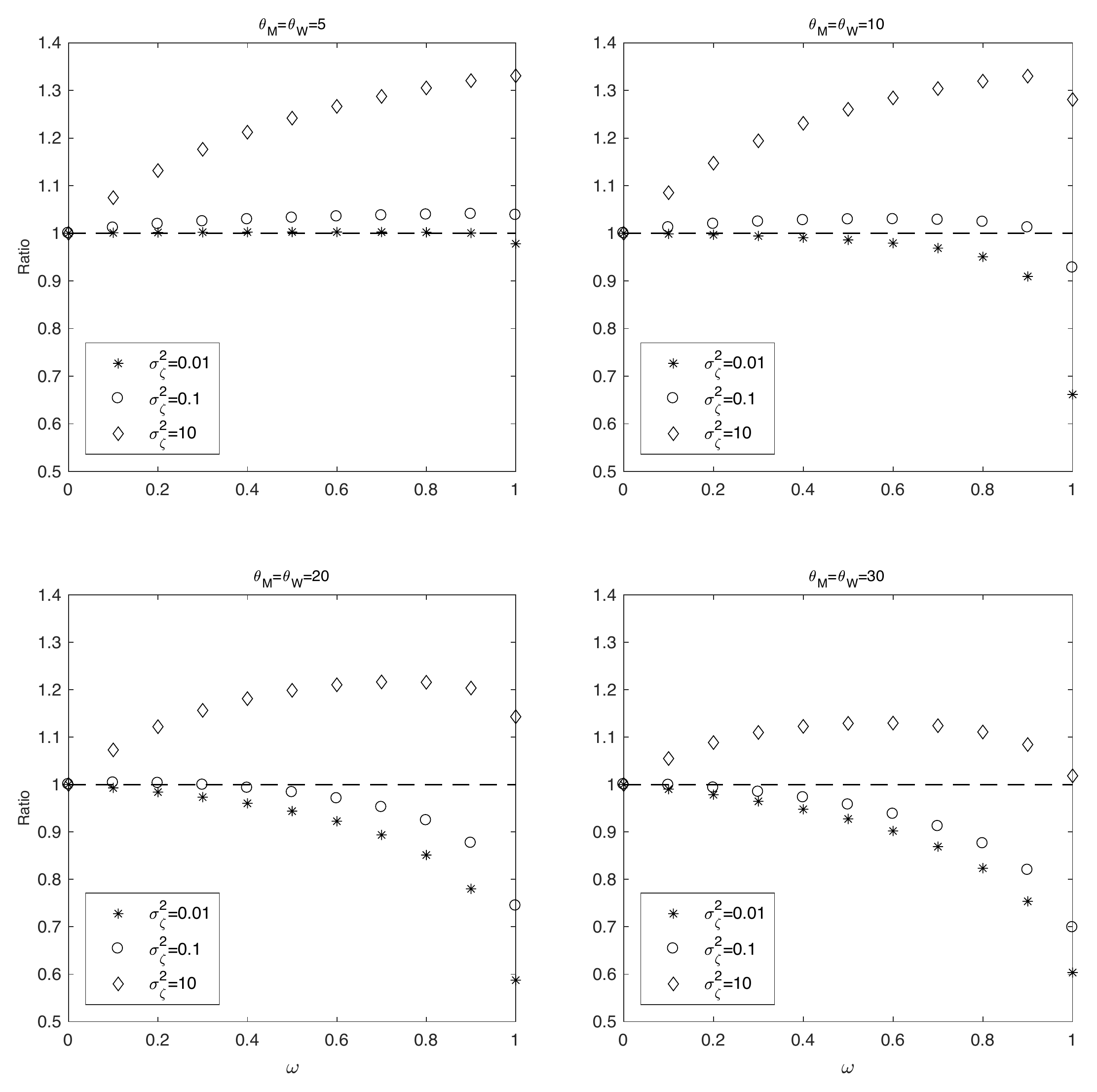}
\caption{The EMSE Ratio \eqref{eq:EMSE_ratio} for $\sigma_\epsilon^2=10$. } \label{fig:CRN_2}
\end{center}
\end{figure}

\section{Conclusions} \label{sec:conclusion}

This paper studies the popular SK metamodel in a new context where the simulation model is inadequate for the real system of interest. We propose the SK-i metamodel  that characterizes both the response surface of the simulation model and its model discrepancy simultaneously. In addition to the two types of uncertainty in the SK metamodel -- one about the response surface of the simulation model and the other about the simulation errors, the SK-i metamodel accounts for the uncertainty about the model discrepancy and the uncertainty about the observation errors of the real system as well, resulting in four types of uncertainty in total. 

Numerous problems arise naturally with regard to, e.g., usage of both  the simulation outputs and the  real data for predicting the real system's performance, estimation of the unknown parameters,  experiment design associated with the simulation model, etc. This paper addresses these problems, thereby laying a mathematical foundation for the SK-i metamodel. In particular, we show both in theory and via numerical experiments that using the augmented data, the SK-i metamodel improves the prediction of the real system substantially in general, relative to the SK metamodel that uses only the simulation outputs and the GPR method that uses only the observations of the real system. 

This paper also provides in-depth analysis of the effect of CRN. In contrast to the known result that CRN is detrimental to the capability of SK to predict the response surface of the simulation model, we show that the effect of CRN on the prediction accuracy of SK-i is complicated by the interplay of the four types of uncertainty involved. We find that the MSE of SK-i may exhibit three kinds of behavior as the CRN-induced correlation grows -- (i) decreasing, (ii) increasing, and (iii) first increasing then decreasing -- depending on a variety of parameters in a nontrivial manner. A case of particular interest is that CRN turns out to be beneficial to prediction if the observation errors of the real system are small enough. 

Uncertainty quantification is an important topic in simulation literature in recent years. Prior research has been focusing on quantifying input uncertainty and its propagation through the simulation model. Model inadequacy apparently represents another type of uncertainty. Our work in this paper suggests that the interplay between various types of uncertainty is highly nontrivial. It is thus of great interest to quantify the overall uncertainty of the real system in the presence of an inadequate simulation model which is subject to input uncertainty itself. We leave it to future investigation.

% Acknowledgments here
\section*{Acknowledgment}
The research is partially supported by Hong Research Grant Council under General Research Fund Project No. 16211417.

\appendix

\section{Proof of Proposition \ref{prop:GPR}.}\label{app:B}

Given $\BFz$,  consider a linear predictor $\widehat{\SFZ}_{\mathrm{GPR}}(\BFx_0) = a + \BFc^\intercal \BFz$. Following a calculation similar to that in the proof of Theorem \ref{theo:BLUP}, we have
\begin{equation}\label{eq:GPR_MSE_proof}
\begin{aligned}
& \MSE^*[\widehat\SFZ_{\mathrm{GPR}}(\BFx_0)] \\
= & [\rho \BFf^\intercal(\BFx_0) \BFbeta + \BFg^\intercal(\BFx_0) \BFgamma -a - \BFc^\intercal (\BFG\BFgamma + \rho \BFF(\ell)\BFbeta)]^2 \\
& + \rho^2\Sigma_{\SFM}(\BFx_0,\BFx_0) +  \Sigma_{\SFW}(\BFx_0,\BFx_0) + \BFc^\intercal [\rho^2\BFSigma_{\BFM(\ell)}+\BFSigma_{\BFW}+\BFSigma_\zeta]\BFc - 2\BFc^\intercal [\rho^2\BFSigma_{\BFM(\ell)}(\BFx_0, \cdot)+ \BFSigma_{\BFW}(\BFx_0, \cdot)],
\end{aligned}
\end{equation}
which is a quadratic function in $\BFc$. Setting the first-order derivative with respect to $\BFc$ to zero yields
\[
2(\rho^2\BFSigma_{\BFM(\ell)} + \BFSigma_{\BFW} +\BFSigma_\zeta)\BFc - 2[\rho^2 \BFSigma_{\BFM(\ell)}(\BFx_0,\cdot) +\Sigma_{\BFW}(\BFx_0,\cdot)] = \BFzero,
\]
which gives the solution
\[\BFc_* = [\rho^2\BFSigma_{\BFM(\ell)} + \BFSigma_{\BFW} +\BFSigma_\zeta]^{-1}[\rho^2 \BFSigma_{\BFM(\ell)}(\BFx_0,\cdot) +\Sigma_{\BFW}(\BFx_0,\cdot)] = \BFV_{\BFtwo\BFtwo}^{-1} \BFC_{\BFtwo}.
\]
Setting the first-order derivative of $\MSE^*[\widehat\SFZ_{\mathrm{GPR}}(\BFx_0)]$ with respect to $a$ to 0 yields
\[
a_* = \rho \BFf^\intercal(\BFx_0) \BFbeta + \BFg^\intercal(\BFx_0) \BFgamma  - \BFc_*^\intercal (\BFG\BFgamma + \rho \BFF(\ell)\BFbeta).
\]
Therefore, the MSE-optimal linear predictor of $ \SFZ(\BFx_0)$ give $\BFz$ is
\[a_* + \BFc_*^\intercal \BFz = \rho\,\BFf^\intercal(\BFx_0) \BFbeta + \BFg^\intercal(\BFx_0)\BFgamma  + \BFC_{\BFtwo}^\intercal \BFV_{\BFtwo\BFtwo}^{-1} [\BFz - \rho \BFF(\ell)\BFbeta - \BFG\BFgamma],\]
proving (\ref{eq:GPR_BLUP}). The unbiasedness of $\widehat{\SFZ}_{\mathrm{GPR}}(\BFx_0)$ is a straightforward result of Proposition \ref{prop:joint_normal}. The optimal MSE can be calculated easily by plugging $a_*$ and $\BFc_*$ into \eqref{eq:GPR_MSE_proof}.

\section{Maximum Likelihood Estimation}\label{app:MLE}

Let $\eta$ denote a generic component of the vector $\Xi = (\rho, \BFbeta,\BFgamma,\tau_{\SFM}^2,\tau_{\SFW}^2,\BFtheta_{\SFM},\BFtheta_{\SFW}, \sigma_\zeta^2)$. The first-order optimality conditions for the MLE is $\frac{\partial \CalL(\Xi)}{\eta} =0 $  for each $\eta$. Such conditions are derived based on explicit calculation using standard results in matrix calculus. Similar results are also given in \cite{ZhangZou16}. We present them here for ease of reference. 

We now derive the first-order derivatives. We write $\BFV = \BFV(\Xi)$ for notational simplicity. By standard results of matrix calculus,
\begin{align*}
\frac{\partial |\BFV|}{\partial \eta} =& |\BFV|\cdot\mathrm{trace}\left(\BFV^{-1} \frac{\partial \BFV}{\partial \eta}\right), \\
\frac{\partial \BFV^{-1}}{\partial \eta} =& - \BFV^{-1}\frac{\partial \BFV}{\partial \eta}  \BFV^{-1}.
\end{align*}
Notice that $\BFV$ does not depend on $\BFbeta$ or $\BFgamma$. For the other components of $\Xi$, we have
\begin{align*}
\frac{\partial \BFV}{\partial \rho} = &
\begin{pmatrix}
\BFzero& \tau_
\SFM^2\BFR_{\BFM(k),\BFM(\ell)}(\BFtheta_{\SFM}) \\[0.5ex]
\tau_{\SFM}^2\BFR^\intercal_{\BFM(k),\BFM(\ell)}(\BFtheta_{\SFM}) & 2\rho\tau_{\SFM}^2\BFR_{\BFM(\ell)}(\BFtheta_{\SFM})
\end{pmatrix}, \\[0.5ex]
\frac{\partial \BFV}{\partial \tau^2_{\SFM}} = &
\begin{pmatrix}
\BFR_{\BFM(k)}(\BFtheta_{\SFM}) & \rho\BFR_{\BFM(k),\BFM(\ell)}(\BFtheta_{\SFM}) \\[0.5ex]
\rho\BFR_{\BFM(k),\BFM(\ell)}^\intercal(\BFtheta_{\SFM}) & \rho^2\BFR_{\BFM(\ell)}(\BFtheta_{\SFM})
\end{pmatrix}, \\[0.5ex]
\frac{\partial \BFV}{\partial \tau_{\SFW}^2} = &
\begin{pmatrix}
\BFzero & \BFzero \\[0.5ex]
\BFzero &  \BFR_{\BFW}(\BFtheta_{\SFW})
\end{pmatrix},\\[0.5ex]
\frac{\partial \BFV}{\partial \sigma_\zeta^2} = &
\begin{pmatrix}
\BFzero & \BFzero \\[0.5ex]
\BFzero &  \BFI_\ell
\end{pmatrix},
\end{align*}
and letting $\BFtheta_{\SFM,p}$ and $\BFtheta_{\SFW,p}$ denote the $p^{\mathrm{th}}$ component of $\BFtheta_{\SFM}$  and $\BFtheta_{\SFW}$, respectively,
\begin{align*}
\frac{\partial \BFV}{\partial \BFtheta_{\SFM,p}} = &
\begin{pmatrix}
\tau^2_{\SFM}\frac{\partial\BFR_{\BFM(k)}(\BFtheta_{\SFM})}{\partial \BFtheta_{\SFM,p}} & \rho\tau_{\SFM}^2\frac{\partial \BFR_{\BFM(k),\BFM(\ell)}(\BFtheta_{\SFM})}{\partial \BFtheta_{\SFM,p} } \\[0.5ex]
\rho\tau_{\SFM}^2\frac{\partial\BFR_{\BFM(k),\BFM(\ell)}^\intercal(\BFtheta_{\SFM})}{\partial  \BFtheta_{\SFM,p}} & \rho^2\tau_{\SFM}^2\frac{\partial\BFR_{\BFM(\ell)}(\BFtheta_{\SFM})}{\partial \BFtheta_{\SFM,p}}
\end{pmatrix}, \\[0.5ex]
\frac{\partial \BFV}{\partial \BFtheta_{\SFW,p}} = &
\begin{pmatrix}
\BFzero & \BFzero \\[0.5ex]
\BFzero & \tau^2_{\SFW}\frac{\partial\BFR_{\BFW}(\BFtheta_{\SFW})}{\partial \BFtheta_{\SFW,p}}
\end{pmatrix}.
\end{align*}
Then, the first derivatives of $\CalL(\Xi)$  are
\begin{align*}
\frac{\partial \CalL(\Xi)}{\partial \BFbeta} =&
\begin{pmatrix}
\BFF(k)  \\[0.5ex]  \rho\BFF(\ell)
\end{pmatrix}^\intercal
\BFV^{-1}
\begin{pmatrix}
\overline{\BFy} - \BFF(k)\BFbeta \\[0.5ex] \BFz - \rho\BFF(\ell)\BFbeta - \BFG \BFgamma
\end{pmatrix},  \\[0.5ex]
\frac{\partial \CalL(\Xi)}{\partial \BFgamma} =&
\begin{pmatrix}
\BFzero  \\[0.5ex]  \BFG
\end{pmatrix}^\intercal
\BFV^{-1}
\begin{pmatrix}
\overline{\BFy} - \BFF(k)\BFbeta \\[0.5ex] \BFz - \rho\BFF(\ell)\BFbeta - \BFG\BFgamma
\end{pmatrix}, \\[0.5ex]
\frac{\partial \CalL(\Xi)}{\partial \rho} =&
-\frac{1}{2}\mathrm{trace}\left[\BFV^{-1}\frac{\partial \BFV}{\partial \rho}\right]
+
\begin{pmatrix}
\BFzero  \\[0.5ex]  \BFF(\ell)\BFbeta
\end{pmatrix}^\intercal
\BFV^{-1}
\begin{pmatrix}
\overline{\BFy} - \BFF(k)\BFbeta \\[0.5ex] \BFz - \rho\BFF(\ell)\BFbeta - \BFG\BFgamma
\end{pmatrix}\\[0.5ex]
&+\frac{1}{2}\begin{pmatrix}
\overline{\BFy} - \BFF(k)\BFbeta \\[0.5ex] \BFz - \rho\BFF(\ell)\BFbeta - \BFG\BFgamma
\end{pmatrix}^\intercal
\left[\BFV^{-1} \frac{\partial \BFV}{\partial \rho}\BFV^{-1}\right]
\begin{pmatrix}
\overline{\BFy} - \BFF(k)\BFbeta \\[0.5ex] \BFz - \rho\BFF(\ell)\BFbeta - \BFG\BFgamma
\end{pmatrix};
\end{align*}
moreover, for $\eta=\tau_{\SFM}^2,\tau_{\SFW}^2,\BFtheta_{\SFM,p},\BFtheta_{\SFW,p},\sigma^2_\zeta$,
\[
\frac{\partial \CalL(\Xi)}{\partial \eta} =
-\frac{1}{2}\mathrm{trace}\left[\BFV^{-1}\frac{\partial \BFV}{\partial \eta}\right]
+ \frac{1}{2}\begin{pmatrix}
\overline{\BFy} - \BFF(k)\BFbeta \\[0.5ex] \BFz - \rho\BFF(\ell)\BFbeta - \BFG\BFgamma
\end{pmatrix}^\intercal
\left[\BFV^{-1} \frac{\partial \BFV}{\partial \eta}\BFV^{-1}\right]
\begin{pmatrix}
\overline{\BFy} - \BFF(k)\BFbeta \\[0.5ex] \BFz - \rho\BFF(\ell)\BFbeta - \BFG\BFgamma
\end{pmatrix}.
\]

% \section{Steady-State Variance of the Waiting Time in the $M/G/1$ Queue}\label{app:MG1}
% Suppose that the service times have the gamma distribution with shape parameter $\alpha$ and rate parameter $\beta$. Then, $\alpha=x^2/\sigma^2$ and $\beta=x/\sigma^2$, since the service time has mean $x$ and variance $\sigma^2$. It is shown in \cite{Blomqvist67} that in steady state,
% \[
% \Var[W_1] = \frac{(\alpha+1)\rho}{12\mu^2\alpha^2(1-\rho)^2}[4(1-\rho)(\alpha+2)+3\rho(\alpha+1)],\]
% and 
% \begin{align*}
% & \sum_{h=0}^\infty \Cov[W_1,W_{1+h}] \\ 
% = & \frac{(\alpha+1)\rho}{24\mu^2\alpha^3(1-\rho)^4}[8\alpha(\alpha+2)-\rho(\alpha-1)(7\alpha+18)+2\rho^2(\alpha-1)(3\alpha+8)-\rho^3(\alpha-1)(\alpha+4)],
% \end{align*}
% where $\mu=1/x$ is the service rate and $\rho=\lambda/\mu = 1/\mu$ is the utilization. Then, a direct calculation yields that
% \[
% \Var[W_1]  + \sum_{h=0}^\infty \Cov[W_1,W_{1+h}] 
% =\frac{(x^2+\sigma^2)[x^3(x+2)+x(x^3-4x^2+5x+4)\sigma^2+(2x^2-8x+9)\sigma^4]}{6x^2(1-x)^4}.\]

\section{Proof of Theorem \ref{theo:CRN_2point}}\label{app:CRN}
A direct calculation yields 
\[
\MSE^*(\omega)= 
\rho^2\tau_{\SFM}^2+\tau_{\SFW}^2-\frac{2\rho^2\tau_{\SFM}^4 r_0^2}{\tau_{\SFM}^2(1+ r_{12})+v(1+\omega)}-r_0^2\left(\rho^2\tau_{\SFM}^2+\tau_{\SFW}^2-\frac{\rho^2\tau_{\SFM}^4(1+r_{12})}{\tau_{\SFM}^2(1+ r_{12})+v(1+\omega)}\right)^2\frac{A(\omega)}{B(\omega)},
\]
where 
\begin{align*}
A(\omega)&=(\tau_{\SFM}^2+v)^2-(\tau_{\SFM}^2r_{12}+ v\omega)^2,\\
B(\omega)&=(\rho^2\tau_{\SFM}^2+\tau_{\SFW}^2+\sigma_\zeta^2)a(r)-\rho^2\tau_{\SFM}^4[(\tau_{\SFM}^2+v)(1+r_{12}^2)-2r_{12}(\tau_{\SFM}^2r_{12} + v\omega )].
\end{align*}
It follows that 
\begin{equation}\label{eq:d_MSE}
\frac{\ud \MSE^*(\omega)}{\ud \omega}=
\frac{-2\rho^2\tau_{\SFM}^2r_0^2H(\omega)}{[\tau_{\SFM}^2(1+r_{12})+v(1+\omega)]^2B^2(\omega)},
\end{equation}
where 
\begin{align*}
C(\omega)=&\rho^2\tau_{\SFM}^2v(1+\omega)+\tau_{\SFW}^2[\tau_{\SFM}^2(1+r_{12})+v(1+\omega)],\\
D(\omega)=&\sigma_\zeta^2A(\omega)+\rho^2\tau_{\SFM}^4v(1-r_{12})(r_{12}-\omega),\\
H(\omega)=&(\tau_{\SFM}^2+v)(\tau_{\SFM}^2r_{12}+v\omega)(1-r_{12})^2C^2(\omega)-(1-r_{12})[\tau_{\SFM}^2(1-r_{12})+r(1-\omega)]C(\omega)D(\omega)-D^2(\omega).
\end{align*}
By \eqref{eq:d_MSE}, it suffices to check the sign of $H(\omega)$ in order to determine the sign of $\frac{\ud \MSE^*(\omega)}{\ud \omega}$. To that end, we calculate $H'(\omega)$ as follows
\begin{align*}
H'(\omega) =&  4\sigma^4_{\zeta}\tau_{\SFM}^2(1+r_{12})A(\omega) \\
&+ \sigma^2_{\zeta}(1-r_{12})[\tau_{\SFM}^2(1+r_{12})+v(1+\omega)] [E(\omega) + F(\omega)] \\ 
 &+ (1-r_{12})^2[\tau_{\SFM}^2(1+r_{12})+v(1+\omega)] G(\omega), 
\end{align*}
where 
\begin{align*}
E(\omega) =& 4(\tau_{\SFM}^2r_{12}+v\omega)[\tau_{\SFM}^2(1-r_{12})+v(1-\omega)]\tau_{\SFW}^2, 
\\ 
F(\omega) =& \rho^2\tau_{\SFM}^2[\tau_{\SFM}^4(1-r_{12}^2)+4v^2\omega(1-\omega)+5\tau_{\SFM}^2v\omega(1-r_{12}) + 4\tau_{\SFM}^2v(1-\omega) + 3\tau_{\SFM}^2vr_{12}], \\
G(\omega) =& [\rho^2v\tau_{\SFM}^2+\tau_{\SFW}^2(v+\tau_{\SFM}^2)]\{\rho^2\tau_{\SFM}^2[\tau_{\SFM}^2(1+r_{12})+3v\omega+v]+\tau_{\SFW}^2[\tau_{\SFM}^2(1+3r_{12})+3v\omega+v]\}.
\end{align*}
It can be seen easily that $E(\omega), F(\omega),G(\omega)\geq 0$ since $0\leq \omega,r_{12}\leq 1$. Moreover, we have $A(\omega)>0$ for $\omega\in(0,1)$, so $H'(\omega)>0$ for $\omega(0,1)$. Hence, there are three cases regarding the sign of $H(\omega)$.  

\textbf{Case (i).} If $H(0)\geq 0$, then for any $\omega\in[0,1]$, $H(\omega)>0$ and thus $\frac{\ud \MSE^*(\omega)}{\ud \omega} <0$. 

\textbf{Case (ii).} If $H(1)\leq 0$, then for any $\omega\in[0,1]$, $H(\omega)<0$ and thus $\frac{\ud \MSE^*(\omega)}{\ud \omega} >0$.

\textbf{Case (iii).} If  $H(0)<0$ and $H(1)>0$, then there exists a unique $\omega^*\in(0,1)$ for which $H(\omega^*)=0$. Moreover, $H(\omega)>0$ for $\omega\in[0,\omega^*)$ and $H(\omega)<0$ for $\omega\in(\omega^*,1]$. 

At last, it is easy to verify that 
\[H(0) \geq 0 \iff \sigma_\zeta^2 \leq  \frac{ \tau_{\SFM}^2\tau_{\SFW}^2 r_{12}(1-r_{12})}{\tau_{\SFM}^2(1-r_{12})+v}\]
and 
\[H(1) \leq 0 \iff \sigma_\zeta^2 \geq \tau_{\SFW}^2r_{12}+v\left(\rho^2+\frac{\tau_{\SFW}^2}{\tau_{\SFM}^2}\right). \]

\bibliographystyle{chicago}
\bibliography{SK_ModErr}
\end{document}